%% file: JeonTongZhao_v8.tex
\documentclass[10pt,twocolumn,twoside]{IEEEtran}
\IEEEoverridecommandlockouts
\usepackage{cite}
\usepackage{amsmath,amssymb,amsfonts}
\usepackage{graphicx}
\usepackage{textcomp}
\usepackage{amsthm}
\usepackage{changepage}
\usepackage{array}
\usepackage{hyperref}
\usepackage{dsfont}
\usepackage{bm}
\usepackage{caption}
\usepackage{soul}
\usepackage{color}

\usepackage{subcaption}
\usepackage{enumitem}
\usepackage{booktabs, tabularx}
\usepackage{multirow}
\usepackage{algorithm}
\usepackage{algpseudocode}
\algrenewcommand\algorithmicrequire{\textbf{Input:}}
\algrenewcommand\algorithmicensure{\textbf{Output:}}
\usepackage{float}
\usepackage{subfiles}
\usepackage{verbatim}
\usepackage{bbm}

\newtheorem{theorem}{Theorem}
\newtheorem{definition}{Definition}

\usepackage{cuted}
\usepackage{stfloats}
\newtheorem{proposition}{Proposition}
\newtheorem{lemma}{Lemma}
\usepackage[font = small, justification = justified,format = plain]{caption}
\newcommand{\chitpr}{\chi_{\text{TPR}}}

\newcommand{\pip}{\pi^+}
\newcommand{\pim}{\pi^-}
\newcommand{\ceilval}[1]{\lceil #1 \rceil}

\def\etal{{\it et al. \/}}

\def\ie{{\it i.e.,\ \/}}

\algrenewcommand\algorithmicindent{0.5em}
\usepackage{xcolor}
\def\BibTeX{{\rm B\kern-.05em{\sc i\kern-.025em b}\kern-.08em
    T\kern-.1667em\lower.7ex\hbox{E}\kern-.125emX}}
\begin{document}

\title{Price-Based Distributed Scheduling of Flexible Demands in Energy Communities}
\author{Minjae~Jeon,~\IEEEmembership{Student Member,~IEEE,}
        Lang~Tong,~\IEEEmembership{Fellow,~IEEE,}
        Qing~Zhao,~\IEEEmembership{Fellow,~IEEE}
    \thanks{Minjae Jeon, Lang Tong,and Qing Zhao are with the School of Electrical and Computer Engineering, Cornell University, USA. This work was supported in part by the National Science Foundation under Grant 2412776, 2419622, and 2603293. }
    }

\maketitle
\begin{abstract}
We study price-based distributed scheduling of flexible demand in an energy community, where a coordinator broadcasts electricity prices and individual households schedule their consumption. Household demand includes deferrable and non-deferrable loads, such as electric vehicle charging with completion deadlines and thermostatically controlled loads. The coordinator transacts with a distribution utility on behalf of community members under the regulated Net Energy Metering tariff. We formulate distributed demand scheduling as a bilevel stochastic dynamic program. The upper level optimizes the coordinator's pricing policy to minimize the community's energy costs subject to operating, revenue-adequacy, and individual-rationality constraints. The lower level involves stochastic dynamic programs that maximize households’ consumption benefits subject to the availability of renewable generation. The computational cost of such a distributed stochastic dynamic program is prohibitive in general. By uncovering the structure of optimal centralized scheduling, we derive Threshold Pricing Rule (TPR)---a simple community pricing policy with linear computational costs for the upper- and lower-level optimizations. Being independent of parameters of the underlying stochastic dynamic program, TPR is robust against modeling uncertainties and is shown to guarantee revenue adequacy for the community and individual rationality for community members. As the community size grows, TPR is shown to be optimal asymptotically.
\end{abstract}

\section{Introduction}
\subfile{sections/Introduction_v8.1}
\section{Problem Formulation} \label{sec:cscheduling}
\subfile{sections/ProblemFormulation_v8}

\section{Optimal Centralized Scheduling}\label{sec:censolution}
\subfile{sections/CenSolution_v8}

\section{Threshold Pricing Rule and Properties}\label{sec:tpr}
\subfile{sections/ThreshPricingRule_v8}

\section{Price-Elastic Non-Deferrable Demand} \label{sec:priceresponsivedemand}
\subfile{sections/PriceResponsive_v8}
\section{Numerical Simulations} \label{sec:numericalsimulation}
\subfile{sections/Numerical_sim_v8}

\section{Conclusions}
\subfile{sections/Conclusion_v8}
\bibliographystyle{IEEEtran}
\bibliography{JeonTongZhaoRef.bib}

\appendices
\section{Proof of Theorems}\label{sec:appendixA}
\subfile{sections/appendixA_proof_v7}

\section{Price-inelastic non-deferrable demand model}
\subfile{sections/appendixB_nondeferrable_load}

\end{document}

%% file: sections/Introduction_v8.1.tex
We study price-based distributed scheduling for networked energy subsystems, in which a coordinator broadcasts optimally designed prices, and each subsystem maximizes its consumption benefit subject to operating constraints. This hierarchical structure admits several closely related formulations: a bilevel optimization problem \cite{dempe2002foundations}, where the coordinator's pricing policy is optimized subject to subsystems' optimal response to coordinator's prices; a Stackelberg game \cite{bacsar1998dynamic}, in which the coordinator is the leader setting the price, and the subsystems are the followers; and a principal-agent problem \cite{laffont2002theory}, in which the principal (coordinator) optimizes organizational objectives, while the agents (subsystems) optimize individual objectives. The general forms of these formulations have been studied extensively. In practice, however, finding the optimal pricing policy for the coordinator and the optimal scheduling policy for the subsystems is challenging and often intractable.

\begin{figure}[t]
 \centering
 \includegraphics[width = 0.8\columnwidth]{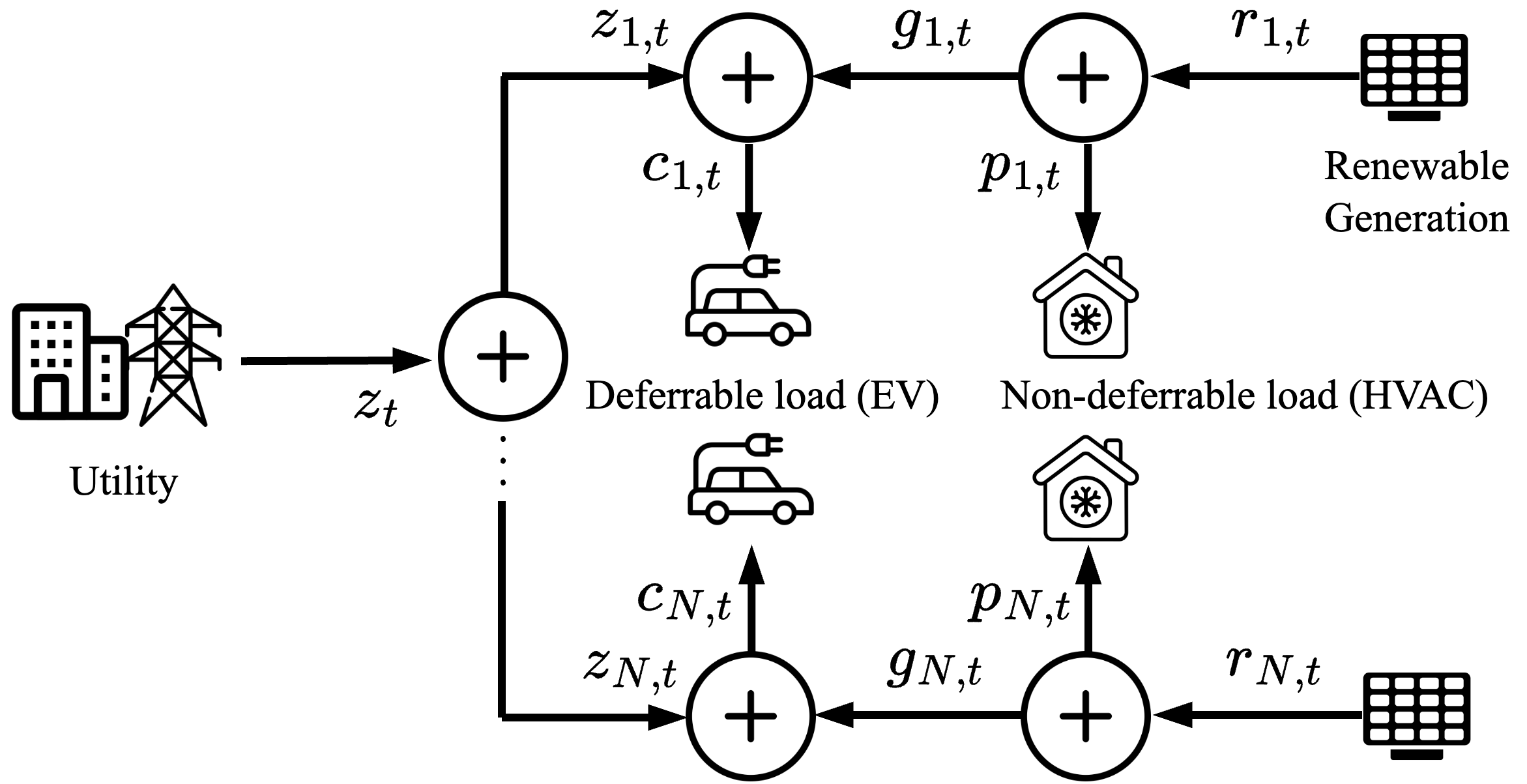}
 \caption{Model of energy community with  $N$ members, where each member has a renewable generation, a deferrable load (EV), and a non-deferrable load (HVAC)}
 \label{fig:energycommunitymodel}
\end{figure}

We consider the distributed scheduling problem in an energy community comprising networked energy subsystems (referred to as households), as shown in Fig.~\ref{fig:energycommunitymodel}. Unlike microgrids, an energy community interacts with its utility only financially. Most energy communities in the U.S. and Europe use a model in which the distribution utility manages the grid, and the community transacts with the utility through a coordinator (a.k.a. operator or aggregator) on behalf of community members under a regulated retail tariff.  The most widely adopted tariff is the Net Energy Metering (NEM): the community is charged at the \textit{retail rate} $\pip$ (\$/kWh) for its net-import and credited at the \textit{compensation rate} $\pim$ (\$/kWh) for its net-export. Typically, $\pip > \pim$, so the export rate is closer to the wholesale price.

While the community faces the utility's NEM tariff, each household faces {\em community prices} set by the coordinator.  The main theme of this work is to find the community welfare-maximizing pricing policy that also ensures  {\em individual rationality}, \ie the community price must be competitive with the utility's NEM tariff.  Otherwise, a member would abandon the community and become a default customer of the utility.  The tension between the coordinator's social objective and members' individual consumption-surplus maximization makes community pricing challenging.

Because individual rationality is essential to the community's stability, a natural question arises: what allows the coordinator, who itself faces the utility's NEM tariff, to offer its members better prices than NEM?

The intuition comes from NEM's price gap $(\pi^-,\pi^+)$ between the retail and compensation rates.  Outside the community, a household selling surplus renewable earns $\pim$ (\$/kWh) below the purchasing retail rate $\pip$ (\$/kWh). Within the community, surplus renewable energy from one household could have been sold to another household.  If the coordinator sets the buying and selling prices in $(\pim, \pip)$, both the selling and buying parties are better off inside rather than outside the community. Thus, sharing surplus renewable energy within the community is the key to competitive community pricing.

We frame the optimal distributed scheduling of flexible demands as optimal community pricing for distributed energy resources and flexible demands. We consider two types of flexible demand. One is deferrable demand with flexible completion time, such as EV charging; the other is non-deferrable flexible demand, possibly price-elastic, such as a thermostatically controlled load. The latter is especially relevant to distributed scheduling because households may have private preferences unknown to the coordinator.

By choice, our energy community model excludes peer-to-peer trading, which has been widely investigated in the literature. One reason to focus on the centralized pricing model is that it leads to simple implementation and strong analytical characterization of pricing properties; achieving the same with peer-to-peer mechanisms is challenging. Another justification is the Second Welfare Theorem. When the community is large and its market is approximately competitive, the equilibrium under the centralized welfare-maximizing price coincides with that of a decentralized (peer-to-peer) market equilibrium. Finally, pricing through a coordinator is widely adopted in practice, while decentralized peer-to-peer trading is rarely used.

\subsection{Related works}
The literature on distributed resource scheduling is vast. Here we restrict the review to those that (i) are applicable to networked energy systems as an energy community, and (ii) involve a coordinator that influences consumption through prices rather than direct scheduling.  Peer-to-peer mechanisms are not reviewed here.

The problem of meeting demand in an energy community has been studied extensively. See a recent review \cite{Barabino23Review} that highlights various energy community models.  The model considered here, as shown in Fig.~\ref{fig:energycommunitymodel}, is a prosumer model with a centralized local market defined in \cite{Barabino23Review}.
Other models approximating  Fig.~\ref{fig:energycommunitymodel} include
\cite{Han&Morstyn&McCulloch:19TPS,
chakraborty2018analysis,
mignoni2023distributed,
alahmed2024dynamic,
li2024decentralized}.

The use of prices as coordination signals for decentralized decision-making has a long history in economics, beginning with Walras’s formalization of competitive equilibrium through prices \cite{Walras:14}, Lange’s proposal that a central planner adjusts prices to implement competitive equilibrium \cite{Lange:36}, and Koopmans’s interpretation of prices as sufficient information for decentralized resource allocation \cite{Koopmans:57}. The present paper is more closely aligned with Arrow’s formulation of internal pricing within organizations \cite{Arrow:77} and Hurwicz’s mechanism design \cite{Hurwicz:73}. Here, the coordinator optimizes the community pricing policy to maximize community welfare while inducing individual optimal decisions consistent with the community welfare objective.

Within the engineering literature, Schweppe \etal made a seminal contribution by using electricity prices to achieve equilibrium in distribution systems \cite{Schweppe&etal:80TPAS}.  Conceptually, we follow the path outlined in \cite{Schweppe&etal:80TPAS}, allowing well-designed prices to guide demand and supply toward a socially optimal outcome.
 In networking technology, Low and Varaiya \cite{Low&Varaiya:93TN}, Kelly \cite{Kelly:97} and Low-Lapsley \cite{Low&Lapsley:99TN} formulated price-based decentralized resource allocation to ensure that individually optimized communicating sources align with the global social welfare objective.  Because the underlying optimization in our case is non-convex, we have to pursue a different path, abandoning the natural choice of using shadow prices from the dual.

The contemporary energy community literature includes  several lines of work on price-based scheduling of shared behind-the-meter (BTM) distributed energy resources (DERs). The most popular line of approach is based on
iterative interactions between the coordinator and individual community members \cite{gan2013optimal,gan2013real,ma2016efficient,Ma&Gupta&Topcu:17TAC,Li&etal:21TSE,mignoni2023distributed,davoudi2025non}, with coordinator sending tentative prices based on tentative scheduling profile and individuals report response to tentative prices until convergence.  The underlying optimizations are static optimization with demand and generation forecasts.  The present paper follows a separate line of non-iterative approaches: the coordinator sets the price by leveraging the structure of the retail tariff, and community members make individual consumption decisions \cite{chakraborty2018analysis,alahmed2024dynamic,li2024decentralized}, where the dynamic NEM \cite{alahmed2024dynamic} is the first to consider the three-zone pricing structure.



%

Missing from the literature are techniques that use price to coordinate distributed {\em dynamic programs}.  The most related is perhaps Hawkins's work on weakly coupled dynamic programs \cite{Hawkins:03Thesis}, where dual variables in the Lagrangian decomposition serve as the pricing signal. Unfortunately, Hawkins’s approach cannot be easily applied to the problem at hand when the state space is uncountable.

\subsection{Summary of results and contributions}
\subsubsection{Optimal Centralized Scheduling Policy Structure} We establish that the optimal centralized scheduling of deadline-constrained demand is a {\em Procrastination Policy} that delays, as much as possible, importing electricity from the grid.  This policy results in a two-threshold policy structure that partitions the axis of community-aggregated renewable energy into three scheduling zones (see Theorem~\ref{thrm1} in Sec.~\ref{sec:censolution}): (i) the {\em net-consuming zone} when the aggregated renewable energy falls below the lower threshold, (ii) the {\em net-producing zone} when the aggregated renewable energy exceeds the upper threshold, and (iii) the {\em net-zero zone} when the aggregated renewable energy is in between.
 Notably, the thresholds and scheduling decisions in the net-producing/consuming zones are given in closed form, thereby isolating the scheduling complexity to the net-zero zone.  This result is a non-trivial extension of earlier work  \cite{alahmed2022net,alahmed2024dynamic} to accommodate the lower-level stochastic dynamic programs.

\subsubsection{TPR: A Novel Community Pricing Policy}
Building on the threshold structure of the optimal centralized scheduling policy, we propose Threshold Pricing Rule (TPR) as the community pricing policy for each scheduling zone. TPR uses the NEM retail rate $\pip$ as the community price in the net-consuming zone, the NEM compensation rate $\pim$  in the net-producing zone, and  the NEM tariff in the net-zero zone. Because prices across all three zones are based only on $\{\pim,\pip\}$, computing TPR only requires trivial updates on the two zone thresholds. Especially significant is that TPR is independent of parameters of the upper level stochastic dynamic program, making TPR robust against model uncertainties.
Further, since $\pi^+$ and $\pi^-$ are known to all households, the TPR price broadcast in each interval requires at most two bits to encode the price change.  Given the community price broadcast by the coordinator, individual households' optimal scheduling can also be implemented in closed form.

\subsubsection{Revenue Adequacy, Individual Rationality, and Asymptotic Optimality} We establish TPR as a feasible, though not optimal, solution to the bilevel stochastic dynamic program.  Specifically, we prove that TPR is revenue adequate, which ensures the community coordinator's revenue is at least as large as its payments to the utility and to community members who exported surplus energy.  We show that  TPR satisfies {\em individual rationality}, which ensures each community member has no incentive to leave the community to become a default customer of the utility. Finally, we show that the TPR is an asymptotically optimal pricing rule under certain traffic conditions, as the community size grows.

\begin{table}[h]
		 \captionsetup{font=scriptsize }
		\caption{NOTATIONS FOR MAJOR VARIABLES}
		\vspace{-1ex}
		\centering
    	\begin{tabularx}{\linewidth}{l l}
			\toprule[0.4mm]
			Symbol & Descriptions
			\\
			\midrule
			$c_{i,t}$ & Charging amount of member $i$
			\\
			$\bar c$ & Maximum per-interval charging amount
			\\
			$d_{i,t}$ & Remaining energy demand of member $i$'s EV
			\\
			$g_{i,t}, \, g_t$ & Net renewable generation of member $i$ and community
			\\
			$m_{i,t}$ & Minimum required charging to meet deadline
			\\
			$M_{i,t}$ & Maximum feasible charging of member $i$'s EV
			\\
$[N]$ & The set of integers $\{1,\cdots, N\}$\\
			$\pmb \mu_i$ & Individual charging policy
			\\
			$\pmb \pi$ & NEM payment parameter
			\\
			$\pi^+, \, \pi^-$ & NEM consumption and compensation rates
			\\
			$\bm p_{i,t}$ & Non-deferrable demand consumption
			\\
			$P_{\pmb \pi}(\cdot)$ & Utility's NEM payment function
			\\
			$P_{\chi} (\cdot)$ & Community payment function
			\\
			$\pmb \psi_t$ & Community price
			\\
			$q(\cdot)$ & Non-completion penalty for EV charging
			\\
			$r_{i,t}, \, r_t$ & Renewable generation of member $i$ and community
			\\
			$\tau_{i,t}$ & Remaining intervals until deadline of member $i$
			\\
			$t$ & Control interval index
			\\
			$T$ & Scheduling horizon
			\\
			$\bar T$ & Maximum charging duration
			\\
			$\chi(\cdot)$ & Pricing rule of community coordinator
			\\
			$\pmb x_{i,t}, \, \pmb x_{t}$ & State of member $i$ and energy community
			\\
		    $z_{i,t}, \, z_t$ & Net consumption of member $i$ and community
			\\
$\mathbbm{1}(\cdot)$ & Indicator function, taking value one if the argument is true \\
& and zero otherwise\\
			\bottomrule[0.4mm]
	\end{tabularx}
 \label{tab:Notations}
\end{table}

%

%% file: sections/ProblemFormulation_v8.tex
This section presents a mathematical model for an energy community and introduces a bi-level stochastic dynamic program for the distributed scheduling problem.  We consider a finite scheduling horizon $T$, indexed by $t = 1, \ldots, T$. The community has $N$ households, indexed by $i = 1, \ldots, N$. With each control interval normalized to the unit length, demand and generation variables are in kilowatt-hours (kWhs).

Major symbols are in Table~\ref{tab:Notations}.  We use boldface lowercase letters, \textit{e.g}. $\pmb a = (a_1, \ldots, a_n)$, to denote column vectors, with 
 $\bm 1$ being the vector of all ones. \\

\subsection{Household demand and energy resources}
Each household has a renewable generator, a deferrable load, and non-deferrable loads as depicted in Fig.\ref{fig:energycommunitymodel}.  Hereafter, we use EV charging as a prototype of deferrable demand with completion deadlines.

	\subsubsection{Deferrable demand with deadline}
We assume that each household has a single EV charger that can charge one EV at a time. Once the EV charging is completed, a new EV demand may arrive according to an independent Bernoulli process with arrival probability $\alpha_{i,t}$.   

Upon arrival at the beginning of interval $t$, two quantities are revealed: the total charging demand $D_i$ and the maximum number of intervals $T_i \le T- t+1$ until deadline.  Let $\mathcal P_{i,t}$ denote the PMF of $(D_i, T_i)$. Arrivals are mutually independent across  households.
	
Let $d_{i,t}$ denote the remaining charging demand for household $i$'s EV at the beginning of interval $t$ and $\tau_{i,t}$ the number of intervals remaining until its deadline. When the charger is unoccupied, we set $(d_{i,t}, \tau_{i,t}) = (0,0)$.  With $\mathcal P_{i,t}(0,0)=1-\alpha_{i,t}$, the state evolution of household $i$'s EV charging  is
	\begin{equation} \label{eq2}
		(d_{i,t+1}, \tau_{i,t+1}) =
		\begin{cases}
			( D_i,  T_i) \sim \mathcal P_{i,t}(\cdot), & (d_{i,t}, \tau_{i,t}) = (0, 0),\\		
(d_{i,t}, \tau_{i,t}) - (c_{i,t}, 1), & \mathrm{otherwise.}
		\end{cases}
	\end{equation}

The charging amount $c_{i,t}$ in each interval is uniformly bounded\footnote{By normalizing each household's demand and per-interval charging amount by their charging capacity, we assume a uniform charging limit without loss of generality.} by the maximum charging amount $\bar c$ and the remaining demand: $c_{i,t} \in [0, \min\{ d_{i,t}, \bar c\}]$. Incompletion of charging is allowed, but any remaining demand $d$ at the deadline incurs a penalty $q(d)$, where  $q$  is strictly increasing and convex.

\subsubsection{Non-deferrable demand} A significant non-deferrable demand is conventional HVAC\footnote{Heating, ventilation, and air conditioning} operation, in which the unit tracks the household temperature and sets HVAC power consumption based on household preferences. Assuming a linear–quadratic (LQ) controller \cite{jia2016dynamic} under standard assumptions on the underlying random processes, the energy consumption admits a Markovian representation.

Let $\bm p_{it}$ be the vector of all non-deferrable demands. To formalize the overall stochastic dynamic optimization, we assume $\bm p_{it}$ is Markovian with transition kernel $h_{i,t}$, independent across $i$:
\begin{equation}
\bm p_{i,{t+1}} \sim h_{i,t}(\cdot \, | \, \bm p_{i,t}).
\end{equation}
Note that price-inelastic non-deferrable demands are must-serve (though random) demands, exogenous to optimizations at the community and individual household levels. The price-elastic demands, on the other hand, are endogenously set by the bilevel optimization.

\subsubsection{Renewable and net-renewable generation}
Household $i$'s renewable generation in interval $t$ is denoted by $r_{i,t}$, assumed to be non-negative and known at the beginning of interval $t$ and modeled as a discrete-time Markov process with uncountable (real) state and state transition distribution $f_{i,t}$:
\begin{equation}
r_{i,t+1} \sim f_{i,t}(\cdot \,| \, r_{i,t}),~~r_t:= \sum_{i \in [N]} r_{i,t},
\end{equation}
where $r_t$ is the community-aggregated renewable. Renewable generations across households are independent.
Combining the exogenous generation and non-deferrable demand processes, we define {\em net-renewable generation} $g_{i,t}$ and its aggregate $g_t$:
	\[
	g_{i,t} := r_{i,t} - {\bm 1}^\top{\bm p}_{i,t},~~g_t:= \sum_{i \in [N]} g_{i,t}.
	\]
\subsection{States of household and community, pricing, and costs}
\subsubsection{States}	Household $i$'s state $\pmb x_{i,t} := (d_{i,t}, \tau_{i,t}, r_{i,t}, \bm p_{i,t})$  includes remaining demand and intervals until deadline, renewable generation, and non-deferrable demand. The community state is denoted as $\pmb x_t := (\pmb x_{1,t}, \ldots, \pmb x_{N,t})$.
	
\subsubsection{Net consumption}	The net household consumption $z_{i,t}$  and net community consumption $z_{t}$ are defined as
	\begin{equation}
		z_{i,t}:= c_{i,t}-g_{i,t};~~z_{ t}:= \sum_{i\in [N]}  z_{i,t}.
	\end{equation}
The household is net-consuming if $z_{i,t}>0$ and {\it net-producing} if $z_{i,t}<0$. The community net-consuming (or net-producing) is similarly defined on $z_t$.

\subsubsection{Community costs under NEM tariff}	
The energy community transacts as a single customer of the utility under the NEM tariff.  In interval $t$, the community cost defined by the NEM payment function $P_{\pmb \pi}$ 
on the community net-consumption $z_t$, parameterized by $\pmb \pi := (\pi^+, \pi^-)$:
\begin{equation} \label{eq:NEMpayment}
	P_{\pmb \pi}(z_t) = \Big[\mathbbm{1}(z_t > 0) \pi^+ + \mathbbm{1}(z_t \le 0) \pi^-\Big]\,z_t,
\end{equation}
where we assume $\pip>\pim$.
	
\subsubsection{Community pricing and household consumption costs}
Within the community, each household's payment is determined by a uniform, time-varying payment function $P_\chi$ applied to its net consumption $z_{i,t}$. The function is parameterized by $\pmb \psi_t$, which the coordinator sets via a pricing policy $\chi: \bm x_t \mapsto \pmb \psi_t$ based on the community state $\bm x_t$.   Knowing the pricing function $P_\chi$ and the broadcast parameter $\pmb \psi_t$, the household schedules its consumption. The cost (or credit) of the net consumption $z_{i,t}$ is given by $P_\chi(z_{i,t};\pmb \psi_t)$.

We assume that the community payment function $P_\chi$  is a monotone increasing function of $z_{i,t}$ with $P_{\chi}(0; \pmb \psi_t) = 0$. Therefore, $P_{\chi}(z_{i,t} ; \pmb \psi_t) < 0$ for $z_{i,t} < 0$, which indicates that the household receives a payment for net production.

\subsection{Distributed scheduling as bilevel optimization}
We now formulate a bilevel stochastic optimization with the upper level for the coordinator to set the community pricing policy $\chi$ and the lower level for the household to follow with the distributed demand scheduling.

\subsubsection{Lower level: Individual Cost Minimization}
For a fixed pricing policy $\chi$, the broadcast prices $\{\pmb \psi_t^{\chi}\}_{t=1}^T$ form a stochastic process whose distribution is induced by $\chi$ and the dynamics of the community state $\pmb x_t$. We assume that households are price takers: household $i$ treats the broadcast price process $\{\pmb \psi_t^\chi\}_{t=1}^T$ as exogenous. Under this assumption, household $i$'s charging cost minimization problem becomes a single-agent stochastic dynamic optimization formulated as:
\begin{equation}\label{eq:householdprob}
	\begin{aligned}
		\min_{ \pmb {\mu}_i  }
		& \quad \mathbb E
		\left[
		\sum_{t=1}^T P_{\chi} (z_{i,t}; \pmb{\psi}_t^\chi) + \mathbbm 1(\tau_{i,t} =1) q(d_{i,t} - c_{i,t})
		\right] \\
		\text{s.t.}
		& \quad (1)-(4), \; c_{i,t} =  \mu_{i,t}   (\pmb x_{i,t}, \pmb \psi_t^\chi),  \; \forall t.
	\end{aligned}
\end{equation}
We denote the optimal policy and optimal cost induced by pricing policy $\chi$ as $\pmb \mu_i^{\chi}$ and $\mathcal C_i^{\chi}$, respectively.

\subsubsection{Community pricing  constraints}
 Given households' responses to the pricing policy, the coordinator's community cost minimization problem is formulated as a stochastic dynamic optimization problem that determines the pricing policy $\chi:  \pmb x_t \mapsto \pmb \psi_t$ subject to two pricing constraints.
 
The first constraint is  {\em Revenue Adequacy}, which ensures that the community is not in deficit in each scheduling interval, \ie the net payment to the utility is no greater than the net revenue collected internally from community members. 
\begin{definition}[Revenue Adequacy] A community pricing policy $\chi : \pmb x_t \mapsto \pmb \psi_t $ is revenue adequate if
	\begin{equation} \label{eq:revenueadequacy}
		P_{\pmb \pi} \bigg( \sum_{i \in [N]} z^\chi_{i,t} \bigg) \le\sum_{i \in [N]} P_{\chi}(z^\chi_{i,t};\pmb \psi_t ), \quad \forall t,
	\end{equation}
where $z^\chi_{i,t}$ is $\chi$-induced net-consumption of $i$ in interval $t$.
\end{definition}

The second constraint is {\em Individual Rationality}, which prevents the community members from leaving the community to become default customers of the utility.  In particular, under pricing policy $\chi$, each household must be better off inside the community than outside, facing the NEM payment function $P_{\bm \pi}$, under which the  optimal consumption is given by
 \begin{equation}\label{eq:NEM}
	\begin{aligned}
\mathcal C_i^{\text{NEM}}:=\min_{\tilde{\pmb {\mu}}_i}\
		& \mathbb E
		\left[
		\sum_{t=1}^T P_{\pmb \pi}(z_{i,t})+\mathbbm 1(\tau_{i,t} =1) q(d_{i,t} - c_{i,t})
		\right] \\
		\text{s.t.}
		& \quad (1)-(4), \; c_{i,t} =  \tilde{\mu}_{i,t}(\pmb x_{i,t}),  \; \forall t.
	\end{aligned}
\end{equation}
The individual rationality constraint of a pricing policy $\chi$ is then defined by comparing each household's optimal cost under $\chi$ against the NEM baseline.
\begin{definition}[Individual Rationality] \label{def:IR}
	The pricing policy $\chi$ satisfies individual rationality if every household's minimum cost under $\chi$ is no greater than that under NEM, \ie
\begin{equation}\label{eq:IR}
\mathbb E \bigg[\sum_{t=1}^T P_{\chi}(z^\chi_{i,t}; \pmb \psi_t) + \mathbbm 1(\tau_{i,t} =1) q(d_{i,t} - c^\chi_{i,t})  \bigg] \le \mathcal C_i^{\text{NEM}}.
\end{equation}
\end{definition}

\subsubsection{Upper level optimization}  The upper-level optimization is
community cost minimization subject to pricing and operating constraints.
\begin{equation}\label{eq:upper}
	\begin{aligned}
		\min_{\chi}
		& \quad \mathbb E
		\left[
		\sum_{t=1}^T P_{\pmb \pi} \bigg (\sum_{i \in [N]} z_{i,t}\bigg) + \sum_{j \in \mathcal J_t} q(d_{j,t} - c_{j,t} )
		\right] \\
		\text{s.t.}
		& \quad (1)-(5), (7)-(9),~ c_{i,t} = \mu_{i,t}^\chi (\pmb x_{i,t}, \pmb \psi_t ), \forall i, t,
	\end{aligned}
\end{equation}
where the set $\mathcal J_t := \{ j \in [N]\, |\, \tau_{j,t} = 1 \}$ represents the set of EVs that have reached their deadlines. We assume that the coordinator knows the deadlines and the remaining demands of all plugged-in EVs and can measure each household's available renewable generation.

\subsection{Model assumptions}
\begin{enumerate}[label=A\arabic*)]
	\item The incompletion penalty function $q$ is a strictly increasing convex function satisfying $q'(0) > \pi^+$.
	\item Initial demand and deadline satisfy $ D_i \le  T_i \bar c$.
	\item Households are rational price-takers, minimizing their consumption costs given the community charging price, treating the community price $(\pmb \psi_t)$ as exogenous. Households truthfully report deadlines for deferrable jobs. 
	\item The coordinator is non-profit and budget-balanced. Any surplus will be returned to households {\it ex post}.
\end{enumerate}
Assumption A1 allows EVs to charge with both purchased energy and renewables. A2 is made without loss of generality.
The price-taking assumption in A3 is a standard microeconomic assumption for competitive markets, introduced to prevent households from strategically influencing future community prices. This is realistic since households lack access to others' information. The truthful revelation of deadlines assumption removes adverse selection by non-truthful revelation of deadlines.   Designing mechanisms such as delay-differentiated pricing to counter such strategic behavior is outside the scope of this work. Under A4, the community pricing policy must generate sufficient revenue to cover the costs.  We do not consider specific mechanisms for distributing the surplus, since such distributions are operated {\it ex post.} Standard solutions include various proportional allocation methods and cooperative game theoretic  mechanisms \cite{Young:94}.

%% file: sections/CenSolution_v8.tex
We begin by  establishing that the optimal scheduling policy admits a threshold structure, which serves as the foundation for deriving the distributed pricing rule in the subsequent analysis.

\subsection{Centralized scheduling problem}
The centralized scheduling problem is a finite-horizon, discrete-time, constrained Markov decision process (MDP) $\mathcal M = (\mathcal S, \mathcal A, \omega, T, \mathbb P)$. The state vector $\pmb x_t \in \mathcal S$ remains as defined previously, while the action vector consists of charging decisions of all $N$ EVs, $\pmb c_t := (c_{1,t}, \ldots, c_{N,t}) \in \mathcal A \subset \mathbb R_+^{N}$. The transition kernel $\mathbb P$ includes $f_{i,t}$, $h_{i,t}$, $\mathcal P_{i,t}$, and the arrival Bernoulli distribution. The stage cost $\omega $ is from the coordinator's upper-level optimization in the distributed problem:
$
\omega(\pmb x_t, \pmb c_t) := P_{\pmb \pi}(z_t) + \sum_{j \in \mathcal J_t} q(d_{j,t} - c_{j,t}).
$

The centralized stochastic dynamic optimization problem optimizes over a sequence of policies $\pmb \nu:= (\nu_1, \ldots, \nu_T)$:
\begin{equation*}
 \begin{aligned}
  \min_{\pmb \nu}
  & \quad \mathbb E
  \bigg[
  \sum_{t=1}^T \omega (\pmb x_t, \pmb c_t)
  \bigg] \\
  \text{s.t.}
  & \quad (1)-(4), \; c_{i,t} = \nu_{i,t}(\pmb x_{t}), \quad \forall i, t.
 \end{aligned}
\end{equation*}

Solving the stochastic dynamic program from Bellman equation is intractable. Instead, we study the structure of the optimal scheduling policy.

\subsection{Procrastination Principle and Threshold Policy}
Consider the special case of a single-member community with $N=1$. This case serves as the benchmark for all community members should they leave the community to become default customers of the utility under the regulated NEM tariff.  More importantly, this case reveals the intuitive procrastination structure of the optimal scheduling policy.

For simplicity, we ignore non-deferrable demand and consider the simpler problem of scheduling deferrable demand under the NEM tariff with retail (purchasing) price $\pip$ (\$/kWh) and compensation (selling) price $\pim$ (\$/kWh) to the utility.

Because $\pim < \pip$, EV charging demand should be met first with locally available renewable energy.  In other words, importing electricity from the utility should be delayed as much as possible until meeting the deadline becomes impossible. We formalize this intuition as the {Procrastination Policy} \cite{Jeon&Tong&Zhao:23CDC}:\\[0.3em]
\underline{\it Procrastination Policy:} In each scheduling interval:
\begin{enumerate}
 \item If the available renewable is more than needed to complete the EV charging demand or greater than the charging capacity, export surplus energy to the grid. The household is a net producer.
 \item If  purchasing electricity from the grid is {\em necessary} to avoid incompletion in the worst case, the demand is in the {\em urgent state}. Import the minimum amount from the grid. In such an interval, the household is a net consumer.
 \item If not urgent, the demand is in the {\em procrastination state.} Charge with the renewable and procrastinate the electricity purchasing decision.  The household is net-zero.
 \end{enumerate}

The Procrastination Policy implies a two-threshold structure that determines when a member needs to export or import electricity to and from the grid, as illustrated in Fig.~\ref{fig:opt_pol_standalone}: (i) the lower threshold $m_{1,t}$ below which grid purchase is necessary when the locally generated renewable is insufficient, and (ii) the higher threshold $M_{1,t}$ above which locally generated renewable energy is more than needed and exporting to the grid is necessary.  In between, all locally generated renewable energy is used for charging.

The two thresholds $(m_{1,t},M_{1,t})$ are easily computed. Recall that, at the beginning of interval $t$, $d_{1,t}$ is the remaining demand to be fulfilled by the deadline, $\bar c$ is the maximum charging amount in a single interval, and $\tau_{1,t}$ is the number of remaining intervals to the deadline.

When the locally generated renewable is small, to avoid the incompletion penalty of unserved demand at the deadline, the minimum charge that must be made in interval $t$ is
\begin{equation}\label{eq:mit}
m_{1,t}(d_{1,t}, \tau_{1,t}) := \max\{d_{1,t} - (\tau_{1,t} - 1) \bar c, 0\}.
\end{equation}
When the locally generated renewable energy is abundant, the maximum amount of renewable energy that can be used for EV charging is
\begin{equation}\label{eq:Mit}
M_{1,t}(d_{1,t}) := \min\{d_{1,t}, \bar c\}.
\end{equation}
Note that both $m_{1,t}$ and $M_{1,t}$ are functions of the state of the individual MDP of the member.  The optimal policy for a standalone community member facing the NEM tariff is:
\begin{equation}\label{eq:optstandalone}
 \tilde \mu_{1,t}(\pmb x_{1,t}) = \begin{cases}
  {m_{1,t}(d_{1,t}, \tau_{1,t}),} & g_{1,t} \le m_{1,t}(d_{1,t}, \tau_{1,t}) \\
  g_{1,t}, & m_{1,t} (d_{1,t}, \tau_{1,t}) < g_{1,t} \\ &\le M_{1,t}(d_{1,t}) \\
  M_{1,t}(d_{1,t}), & M_{1,t}(d_{1,t}) < g_{1,t},
 \end{cases}
\end{equation}
which achieves $\mathcal C_1^{\mathrm{NEM}}$. The optimal decisions for more general cases are derived in \cite{Jeon&Tong&Zhao:23CDC}.
\begin{figure}[t]
 \centering
 \includegraphics[width=0.8\columnwidth]{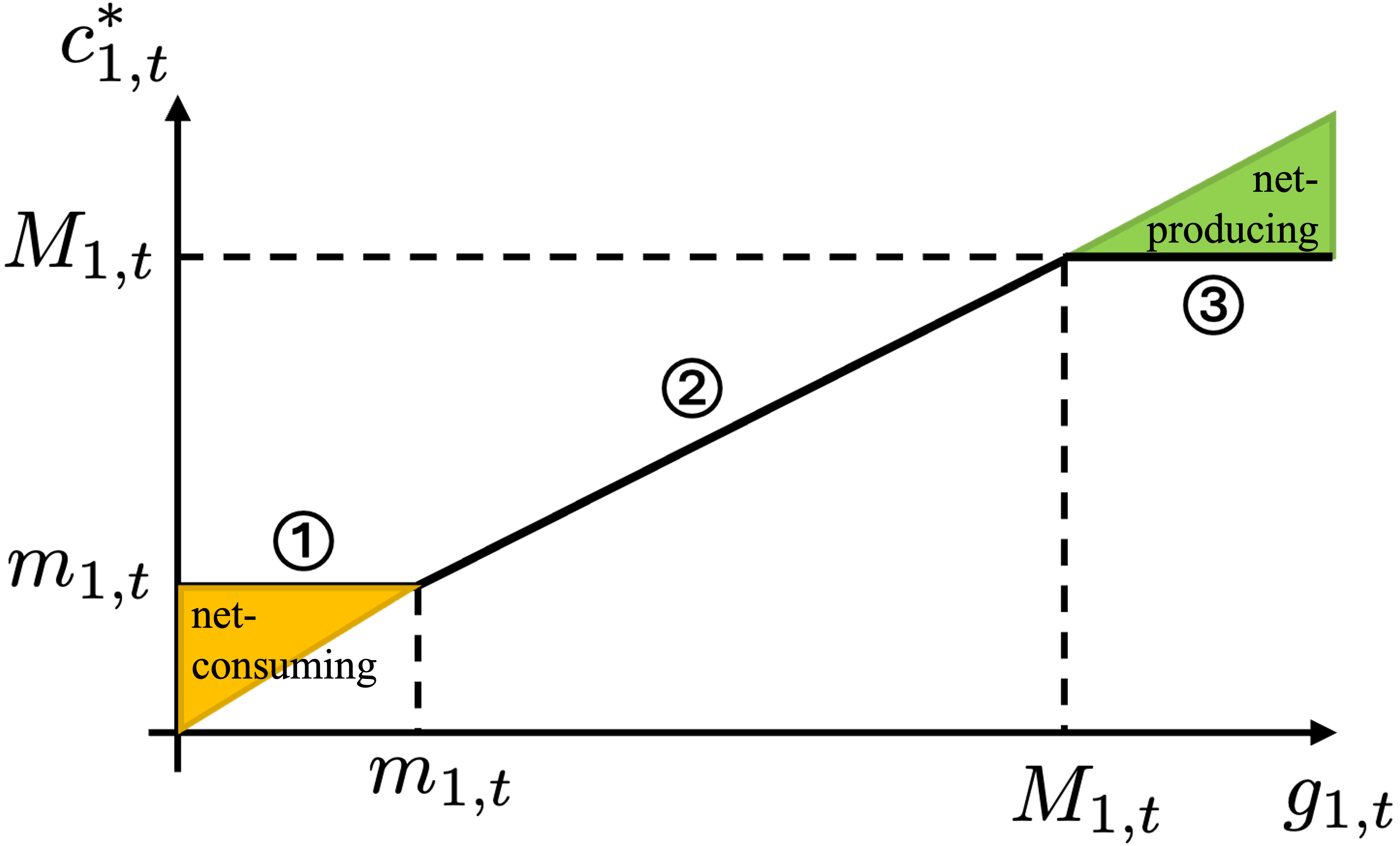}
 \caption{The optimal charging policy of a stand-alone customer with three zones on the net-renewable generation:  \textcircled{1} net-consuming (yellow) zone where the customer net-imports from the grid, \textcircled{2} net-zero zone where the customer uses only local renewable, and \textcircled{3} net-producing (green) zone where the customer is a net producer.}
 \label{fig:opt_pol_standalone}
\end{figure}

\subsection{Optimal threshold structure of centralized decision}
We now consider the optimal centralized scheduling for an energy community with size $N$.

\begin{theorem}[Optimal threshold policy]\label{thrm1}
	The optimal charging policy $\nu_t^*: \pmb x_t \mapsto \pmb c_t^*=(c_{1,t}^*, \ldots, c_{N,t}^*)$ is a threshold policy defined by two thresholds, $M_t(\pmb d_t)$ and $m_t(\pmb d_t, \pmb \tau_t)$, on the aggregate net renewable generation $g_t$, and the net-zero zone policy $\pmb \rho_t (\pmb x_t) = (\rho_{1,t}(\pmb x_t), \ldots, \rho_{N,t} (\pmb x_t))$:
	\begin{equation}
		c_{i,t}^* =
		\begin{cases}
			m_{i,t} (d_{i,t}, \tau_{i,t}) & g_t \le m_t (\pmb d_t, \pmb \tau_t)\\
			\rho_{i,t}(\pmb x_t), & m_t (\pmb d_t, \pmb \tau_t) < g_t \le M_t (\pmb d_t) \\
			M_{i,t}(d_{i,t}), & M_t(\pmb d_t) < g_t
		\end{cases}\label{eq:cen_opt}
	\end{equation}
	where $\pmb \rho_t(\pmb x_t)$ satisfies $\sum_{i \in [N]} \rho_{i,t}(\pmb x_t) = g_t$. The two thresholds are defined as
	\begin{equation*}
		m_t(\pmb d_t, \pmb \tau_t) := \sum_{i \in [N]} m_{i,t}(d_{i,t}, \tau_{i,t}), \; M_t(\pmb d_t) := \sum_{i \in [N]} M_{i,t}(d_{i,t}).
	\end{equation*}
\end{theorem}

The proof is in Appendix~\ref{sec:appendixA} by an interchange argument. The optimal centralized policy extends the stand-alone procrastination strategy from the individual to the community level, with the same net-consuming, net-zero, and net-producing zones defined by two thresholds $(m_t,M_t)$ on the aggregate surplus renewable generation $g_t$. The lower threshold $m_t$ is the aggregate of individual lower thresholds $m_{i,t}$,  and the upper threshold $M_t$ is the sum of individual upper thresholds $M_{i,t}$ from (\ref{eq:mit})-(\ref{eq:Mit}).

Note that, in the net-consuming zone, the community imports  the minimum amount from the grid for EVs in an urgent state.  In the net-producing state, all EVs charge to their maximum per-interval level. In both cases, the optimal import/export quantities are given in closed form, and the computational cost is linear in the size of the community.

In the intermediate net-zero zone, charging is allocated so that aggregate charging matches aggregate generation; the community neither imports nor exports electricity. Unfortunately, the optimal allocation of aggregated renewables to individual households requires solving the Bellman equation, which is generally intractable.

The most striking features of centralized scheduling are the two-threshold structure and the closed-form expressions for the thresholds and the export/import quantities.  Next, we will exploit these properties to obtain a price-based distributed scheduling solution.

%% file: sections/ThreshPricingRule_v8.tex
This section focuses on the community pricing policy of the coordinator. Given that directly solving the bilevel optimization seems intractable, our approach is to leverage the structural property of the optimal centralized scheduling solution by adopting a three-zone scheduling structure and finding prices that directly induce the optimal responses from individual households.

It turns out that finding the right prices to induce the optimal response is easy for the net-consuming and net-producing zones.  In the net-zero zone, unfortunately, it is more challenging, and solving a dynamic program is necessary. Trading optimality in the net-zero zone for simplicity, we adopt the NEM pricing rule there.  We then establish that TPR is feasible for the bilevel optimization in Proposition~\ref{prop:IR}-\ref{prop:RA}.

\subsection{Threshold Pricing Rule (TPR)}
TPR has the same structure as the utility's NEM tariff; it is uniform, piecewise linear, and dynamic with time-varying parameters $\pmb \psi_t =(\psi_t^+, \psi_t^-)$.  Specifically, for member $i$ with net consumption $z_i$,  the payment is given by, as in (\ref{eq:NEMpayment})
\begin{equation}\label{eq:community_payment}
	P_{\pmb \psi_t}(z_{i}) = \big[\mathbbm{1}(z_i > 0) \psi_t^+ + \mathbbm {1} (z_i \le 0) \psi_t^- \big]z_{i},
\end{equation}
where $\psi_t^+$ and $\psi_t^-$ represent the consumption and compensation rates within the community, respectively.

\begin{definition}[Threshold Pricing Rule] The TPR pricing policy $\chi_{\text{\normalfont TPR}}: \pmb x_t \mapsto \pmb \psi_t$ sets the prices $\pmb \psi_t$ using the thresholds of the optimal centralized policy:
\begin{equation}
	\pmb\psi_t = 	
	\begin{cases}
		\pmb \pi^+ := (\pi^+, \pi^+), & g_t \le m_t(\pmb d_t, \pmb \tau_t)  \\
		\pmb \pi = (\pi^+, \pi^-), & m_t (\pmb d_t, \pmb \tau_t)  < g_t \le M_t (\pmb d_t) \\
		\pmb \pi^- := (\pi^-, \pi^-), & M_t (\pmb d_t) < g_t.
	\end{cases}
\end{equation}
\end{definition}

TPR is simple. It sets prices according to the zonal structure of  the optimal centralized scheduling policy.  It sets the price to $\pip$ in the net-consuming zone,  $\pim$ in the net-producing zone, and standard NEM in the net-zero zone.  Most significant, perhaps,  is that TPR is independent of the model parameters of the upper and lower level MDPs, making TPR robust to the underlying modeling assumptions.

Note that the computation involved in TPR is trivial; with $(\pim,\pip)$ known, the computational cost comes from computing thresholds and is linear in community size $N$. Note also that TPR is linear in the net-consuming/producing zones and piecewise linear (NEM) in the net-zero zone.

\subsection{Optimal household response to TPR}
TPR is designed so that the community price can induce procrastination in individual members as an outcome of lower-level optimization. The following proposition formalizes the solution to the household's cost minimization problem (\ref{eq:householdprob}) under TPR.

\begin{proposition}[Optimal household response to TPR]
	Given the community pricing rule $\chi_{\text{\normalfont{TPR}}}$, member $i$'s optimal charging decision is myopic, depends only on the current broadcast price, and is given by $\mu_i^{\chi_{\text{\normalfont{TPR}}}} $:
	\begin{equation*}
		\mu_i^{\chi_{\text{\normalfont{TPR}}}} (\pmb x_{i,t}, \pmb \psi_t)=
		\begin{cases}
			m_{i,t}(d_{i,t}, \tau_{i,t}), & \pmb\psi_t = \pmb \pi^+ \\
			\tilde{\bm \mu}_{i,t} (\pmb x_{i,t}), & \pmb\psi_t = \pmb \pi \\
			M_{i,t}(d_{i,t}), &\pmb \psi_t = \pmb \pi^-
		\end{cases}.
	\end{equation*}
	\label{prop1}
\end{proposition}
The proof is in Appendix~\ref{sec:appendixA}. When the community rate is $\pmb{\pi}^+$, household's optimal charging decision follows the procrastination policy in the net-consuming zone: urgent households $(m_{i,t} > 0)$ charge the minimum required amount to avoid unnecessary purchases at the high rate $\pip$, while procrastinating households $(m_{i,t} = 0)$ sell their renewables to other urgent households at the high rate $\pip$ to maximize the profit. Together, this replicates the centralized optimal decision in the net-consuming zone. When the community rate is $\pmb {\pi}^-$, households maximize their charging to exploit the low purchase rate $\pi^-$, rather than selling renewables at the same low rate, again replicating the centralized optimal decision in the net-producing zone.

For $\pmb \psi_t = \pmb\pi^+$ and $\pmb \psi_t = \pmb \pi^-$, the distributed scheduling under TPR recovers the centralized optimal decision of Theorem~\ref{thrm1}. Given state $\pmb x_t$, if the aggregate net renewable generation $g_t$ lies in the net-consuming zone $(g_t \le m_t)$ or the net-producing zone $(g_t \ge M_t)$, then each member's myopic best response to the broadcast price coincides with the centralized optimal action. In other words, members acting in their own interest under TPR make the same decisions that the coordinator would have imposed centrally, which shows that TPR has the incentive-alignment feature in the two outer zones.


\subsection{Feasibility of TPR: Individual Rationality}
Individual rationality as defined in (\ref{eq:IR}) is part of the upper-level optimization constraint, necessary to ensure the stability of the community as a coalition. We establish this property for TPR in Proposition~\ref{prop:IR}.

\begin{proposition}[Individual rationality of TPR] \label{prop:IR}
	Under $\chi_{\text{\normalfont{TPR}}}$, every household's optimal charging cost 	$\mathcal C_i^{\chi_\text{\normalfont TPR}}$ is no greater than their stand-alone charging cost $\mathcal C_{i}^{\text{\normalfont NEM}}$:
	\begin{equation*}
		\mathcal C_i^{\chi_\text{\normalfont TPR}} \le \mathcal C_{i}^{\text{\normalfont NEM}}, \quad \forall i.
	\end{equation*}
\end{proposition}

With proof in Appendix~\ref{sec:appendixA}, we offer simple intuitions why the community pricing of TPR is more favorable than NEM outside the community. Note that TPR is defined separately in three zones.  In the net-zero zone, TPR matches with the NEM, and the optimal response policy is identical to that outside the community.   In the net-consuming and net-producing zones, TPR prices are linear, unlike NEM.

In the net-consuming zone $(\pmb \psi_t = \pmb \pi^+)$, a procrastinating member defers charging and sells surplus renewable energy to urgent members at the price $\pip$ (Member 3 in Fig.~\ref{fig:IR}a). Since $\pip$ exceeds the stand-alone compensation rate $\pim$, this member earns strictly more than outside the community, while urgent members pay no more than the outside consumption rate $\pip$ (Members 1 and 2 in Fig.~\ref{fig:IR}a).

In the net-producing zone $(\pmb \psi_t = \pmb \pi^-)$, surplus renewable energy is shared with urgent members at the price $\pim$ (Member 2 in Fig.~\ref{fig:IR}b). Since $\pim$ is below the stand-alone consumption rate $\pip$, the urgent members pay strictly less, while procrastinating households receive the same compensation rate $\pim$ as they would from the grid (Households 1 and 3 in Fig.~\ref{fig:IR}b). In both zones, the household whose net consumption has the opposite sign to the community aggregate net consumption captures a strict benefit from the price differential—highlighted by red arrows in Fig.~\ref{fig:IR}—while all others are left no worse off.

\begin{figure}[t]
	\centering
	\includegraphics[width=\columnwidth]{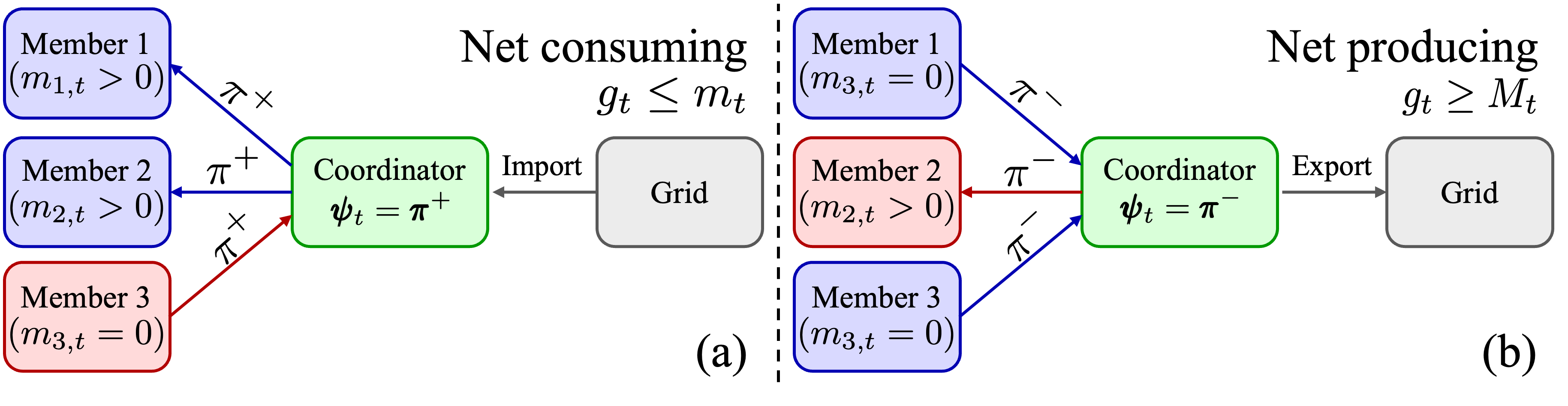}
	\caption{Illustration of individual rationality under the TPR. Arrows represent the net consumption $z_{i,t}$ and $z_t$. (a) Net consuming. (b) Net producing}
	\label{fig:IR}
\end{figure}

\subsection{Feasibility of TPR: Revenue Adequacy}
The revenue adequacy condition is necessary to ensure community members do not incur out-of-market costs. The Proposition~\ref{prop:RA} below establishes this property formally.  The key intuition is that sharing surplus renewable energy within the community is more profitable than selling it to the utility.

\begin{proposition}[Revenue adequacy of TPR] \label{prop:RA}
	Let $\pmb \psi_t^{\chi_{\text{\normalfont{TPR}}}}$ denote the price determined by {\normalfont TPR}, and let $c_{i,t}^{\text{{\normalfont TPR}}} = \mu_i^{\chi_{\text{\normalfont{TPR}}}} (\pmb x_{i,t}, \pmb \psi_t)$ denote the optimal individual action under {\normalfont TPR}. Then, {\normalfont TPR} satisfies the revenue adequacy for all $t$:
		\begin{align*}
	\sum_{i \in [N]} P_{\pmb \psi_t^{\chi_{\text{\normalfont{TPR}}}}}& \big(
	c_{i,t}^{\text{{\normalfont TPR}}} - g_{i,t})\ge P_{\pmb \pi} \bigg(
	\sum_{i \in [N]}  c_{i,t}^{\text{{\normalfont TPR}}} - g_{i,t}			
	\bigg).
\end{align*}
\end{proposition}																													

{\it Proof:} If TPR price is $\pmb \pi^+$, community is importing since TPR induced action satisfies $\sum_i c_{i,t}^{\text{TPR}} > g_t $, and community faces $\pip$ for the net-consumption. As both prices are identical, the power balance implies revenue balance, and there is no surplus. The same arguments applies also for TPR price $\pmb \pi^-$.

For TPR price $\pmb \pi$, TPR is NEM. Let $Q^+$ and $Q^-$ be the total net-consumption and net-production within the community respectively. Suppose $Q^+>Q^- $. With $\pi^+>\pi^-$, the community surplus S is
\[
\mbox{S}=(\pi^+Q^+-\pi^-Q^-)-\pi^+(Q^+-Q^-) \ge 0.
\]
The same argument applies to $Q^+<Q^-$. \hfill $\Box$

\subsection{Asymptotic Optimality}
TPR is not optimal. The suboptimality arises in the net-zero zone, where TPR adopts the NEM pricing rule rather than inducing the optimal centralized allocation. This section shows that, as the community grows, the net-zero zone is visited with vanishing probability under light-traffic conditions, and TPR becomes asymptotically optimal.

\subsubsection{The asymptotic regime}
We consider a homogeneous community in which when household $i$'s charger is unoccupied at the beginning of interval $t$, an EV arrives with probability $\alpha$, independently across households. We further assume that the charging duration of the new arrival $T_i$ is limited by $\bar T$.

For the renewables, we model the individual net renewable generation $g_{i,t}$ as i.i.d.\ random variables across members and stationary in time with marginal mean $\mathbb{E}[g_{i,t}] = \theta_g$. By the strong law of large numbers, as the size of the community grows to infinity, the sample mean of the individual members' $g_{i,t}$ converges to $\theta_g$ almost surely. 


\subsubsection{Asymptotic optimality of TPR}
Let $ V^N(x_0)$ and $V^{N,  \chi_{\text{\normalfont{TPR}}}} (x_0)$ be the optimal centralized scheduling cost and the expected community cost under the TPR for a community of size $N$, respectively. The following theorem establishes the asymptotic optimality of TPR under the light traffic regime.

\begin{theorem}(Asymptotic optimality of TPR in light traffic regime) \label{thrm:asymptoticTPR}
		Consider the homogeneous Bernoulli arrival model with arrival probability $\alpha$. Further, assume that the renewables are i.i.d. random across members and stationary in time with mean $\theta_g$. Then, the TPR is asymptotically optimal as $N \to \infty$, if $\alpha < \theta_g / (\bar c \bar T)$:
	\begin{equation*}
		\lim_{N \to \infty} \; \frac 1 N \bigg( V^{N, \chi_{\text{\normalfont{TPR}}}} (x_0) - V^N(x_0)  \bigg) = 0.
	\end{equation*}
\end{theorem}

The sketch of proof is in Appendix~\ref{sec:appendixA}. The theorem implies that if the arrival rate grows slower than the community's aggregate renewable, the distributed scheduling using the TPR becomes asymptotically optimal. The light traffic condition $\alpha \le \theta_g /( \bar c \bar T)$ admits a natural interpretation in terms of job scheduling, where $ N\theta_g /\bar c$ represents the number of EVs that can be simultaneously served by the community's aggregate generation, playing the role of the service rate. It is inversely proportional to the maximum charging duration $\bar T$ because a longer maximum charging duration means that more EVs will be present in the community simultaneously. Under this condition, the community is in the net-producing zone---where $\pmb \pi^-$ induces the centralized optimal action---with probability approaching one as $N \to \infty$.

The homogeneity of arrival probability is assumed for simplicity. If arrival probabilities $\alpha_{i,t}$ vary across households, we can establish asymptotic optimality with the modified light traffic condition of $\sup_{i,t} \alpha_{i,t} < \theta_g/(\bar{c}\,\bar{T})$.

%% file: sections/PriceResponsive_v8.tex
The preceding analysis treats the non-deferrable demand $\bm p_{i,t}$ as price-inelastic. We now allow it to be price-elastic, endowing each household with a private utility $U_{i,t}(\bm p_{i,t})$ so that consumption responds to the community price. The household's objective then becomes surplus maximization rather than cost minimization. Under a modified rule $\tilde \chi_{\text{TPR}}$ that preserves the three-zone NEM-based structure of TPR, the theoretical guarantees of the price-inelastic case carry over in surplus terms: $\tilde \chi_{\text{TPR}}$ retains revenue adequacy, achieves individual rationality with respect to each member's surplus, and remains asymptotically optimal, now converging to the community's social welfare optimum. For notational simplicity, we consider scalar non-deferrable demand in this section; the vector extension is immediate. 

We take the utility $U_{i,t}(p_{i,t})$ to be differentiable, concave, and time-varying, and private to member $i$, who reports only its realized consumption level to the coordinator. Given the community price, it is scheduled interval-by-interval to maximize individual surplus:
\begin{equation}
	\max_{p_{i,t}\in\mathbb{R}_+}\; U_{i,t}(p_{i,t}) - P_{\pmb \psi_t}(p_{i,t} - r_{i,t}).
	\label{eq:elastic-nd}
\end{equation}
Let  $\partial U^{-1}_{i,t}$ denote inverse marginal utility function. For community prices $\pmb \pi^+$ and $\pmb \pi^-$, the optimal non-deferrable consumption levels are given by $p_{i,t}(\pmb \pi^+) = \partial U_{i,t}^{-1}(\pi^+)$ and $p_{i,t}(\pmb \pi^-)=\partial U_{i,t}^{-1}(\pi^-)$, respectively. For simplicity, we assume the utility function is such that $\partial U_{i,t}^{-1}(\pi^+) > 0$ and $\partial U_{i,t} ^{-1}(\pi^-) > 0$. When the price is $\pmb \psi_t = \pmb \pi$, the optimal non-deferrable demand consumption level is determined by a threshold rule based on $r_{i,t}$
\begin{equation}
	p_{i,t}(\pmb \pi, r_{i,t}) =
	\begin{cases}
		\partial U_{i,t}^{-1}(\pi^+), & r_{i,t} < \partial U_{i,t}^{-1}(\pi^+),\\[2pt]
		r_{i,t}, & \partial U_{i,t}^{-1}(\pi^+) \le r_{i,t} \le \partial U_{i,t}^{-1}(\pi^-),\\[2pt]
		\partial U_{i,t}^{-1}(\pi^-), & r_{i,t} > \partial U_{i,t}^{-1}(\pi^-),
	\end{cases}
	\label{eq:elastic-nd-threshold}
\end{equation}
Note that by the concavity of $U_{i,t}$, $\partial U_{i,t}^{-1}$ is monotone decreasing and non-deferrable consumption level satisfies $p_{i,t}(\pmb \pi^+) \le p_{i,t}(\pmb  \pi, r_{i,t}) \le p_{i,t}(\pmb \pi^-)$.

Because the served non-deferrable demand now depends on the broadcast price, so does the member's net renewable generation $g_{i,t}(\pmb \psi) := r_{i,t} - p_{i,t}(\pmb \psi)$. TPR cannot be applied directly because of the dependency of non-deferrable demand and net-renewable on the price.

A simple modification prevails, fortunately. In particular, we assume that at the beginning of every interval, each member reports its non-deferrable demand consumption levels for three zones with the current renewable generation level. This does not require the coordinator to know members private utility function.  Given this information, the coordinator applies the pricing rule $\tilde \chi_{\text{TPR}}$:
\begin{equation}
	\pmb \psi_t = \tilde\chi_{\text{TPR}}(\pmb x_t) =
	\begin{cases}
		\pmb \pi^+, & \textstyle\sum_{i\in[N]} g_{i,t}(\pmb \pi^+) < m_t(\pmb d_t, \pmb\tau_t),\\[2pt]
		\pmb \pi^-, & \textstyle\sum_{i\in[N]} g_{i,t}(\pmb \pi^-) > M_t(\pmb d_t),\\[2pt]
		\pmb \pi,   & \text{otherwise}.
	\end{cases}
	\label{eq:elastic-tpr}
\end{equation}
Since the community price is exogenous to each member and Proposition~\ref{prop1} holds on a realization basis for the surplus renewable generation level, its result carries over to household's optimal deferrable demand decision. For $c_{i,t}^{\text{TPR}} = \mu_i^{\chi_{\text{TPR}}} (\pmb x_{i,t}, \pmb \psi_t)$ and optimal non-deferrable demand level $p_{i,t}^*$, define member $i$'s optimal surplus under $\tilde\chi_{\text{TPR}}$ as
\begin{equation}
	S_i^{\tilde\chi_{\text{TPR}}} := \sum_{t=1}^{T} \mathbb{E}\!\left[
	U_{i,t}\big(p_{i,t}^*\big) - P_{\pmb \psi_t}\big(p_{i,t}^* + c_{i,t}^{\text{TPR}}- r_{i,t}\big)
	\right].
	\label{eq:elastic-surplus}
\end{equation}
Let $S_i^{\mathrm{NEM}}$ denote the optimal surplus of the stand-alone customer with default utility. With individual surplus as defined above, $\tilde\chi_{\text{TPR}}$ is admissible: satisfies individual rationality in terms of surplus and satisfies revenue adequacy.

\begin{proposition}[Feasibility of $\tilde\chi_{\text{\normalfont TPR}}$]
	\label{thm:surplus-ir}
	Under $\tilde\chi_{\text{\normalfont TPR}}$, every community member's optimal surplus is no less than its stand-alone optimal surplus:
	$S_i^{\tilde\chi_{\text{\normalfont TPR}}} \ge S_i^{\mathrm{NEM}}$ for all $i$. Further, $\tilde\chi_{\text{\normalfont TPR}}$ satisfies revenue adequacy for all $t$.
\end{proposition}

The proof is in \cite{jeon2026optimal}. Asymptotic optimality likewise extends to $\tilde \chi_{\text{\normalfont{TPR}}}$ with respect to community surplus and the benchmark is the centralized welfare optimum. The light traffic condition is modified to $\alpha < \theta_g(\pim)/(\bar c \bar T)$, where $\theta_g(\pim) := \mathbb E[r_{i,t} - p_{i,t}(\pim)]$ for homogeneous households; see \cite{jeon2026optimal}.

%% file: sections/Numerical_sim_v8.tex
In this section, we present numerical simulations of the proposed distributed scheduling algorithm. We first describe the simulation setting in Section~\ref{sec:simulation_setting}, followed by a numerical demonstration of asymptotic optimality and comparisons of community and individual cost  savings.

\subsection{Simulation setting }\label{sec:simulation_setting}
We considered an energy community with $N$ members; each had a BTM solar generator and an EV charger. The scheduling horizon was 24 hours with 15-minute control intervals.

\subsubsection{Evaluation datasets}
Simulations used real-world and synthetic data. EV charging demand was drawn from the Adaptive Charging Network dataset at Caltech \cite{lee_acndata_2019}. Renewable generations used i.i.d. lognormal profiles to demonstrate asymptotic optimality.
For other studies, real-world BTM solar profiles from Pecan Street's New York residential customers were used \cite{pecanstreetdata}.

\subsubsection{EV demand}
EV arrivals followed a homogeneous Bernoulli process; charging durations were drawn from a truncated Gaussian with mean 5 hours; the charger's maximum rate was 7.2 kW, which is typical for Level-2 chargers.

\subsubsection{Energy supply and non-deferrable demand}
The coordinator either purchased from the grid or used renewables to meet community demand, selling any surplus to the grid when no EVs were present. We set $p_{i,t} = 0$ to demonstrate the benefit of renewable sharing in EV charging. The NEM rates were $(\pi^+, \pi^-) = (0.5$ \$/kWh, $0.2$ \$/kWh$)$.

\subsection{Asymptotic approximation of optimality } This experiment aimed to validate the asymptotic optimality established in Theorem~\ref{thrm:asymptoticTPR}. Since computing the optimal centralized scheduling was intractable, we compared TPR against the  Oracle policy that solved an open-loop optimization with perfect knowledge of all random variables, which served as a lower bound on the cost of the optimal policy. In addition, we also compared two suboptimal centralized scheduling: least-laxity-first (LLF) and model predictive control (MPC).

Fig~\ref{fig:asymptoticopt} shows the gap of the average community cost per household relative to the Oracle policy on a log scale. The gap of the distributed scheduling based TPR exhibited exponential decay to zero as the community size grew, which indicates that the convergence shown in Theorem~\ref{thrm:asymptoticTPR} could be exponential. In contrast, the optimality gaps of LLF and MPC persisted as the community size increased. It is remarkable that the simple TPR-based distributed scheduling significantly outperformed far more complex centralized scheduling in both scheduling and computational costs. 

\begin{figure}[t]
 \centering
 \includegraphics[width=0.8\columnwidth]{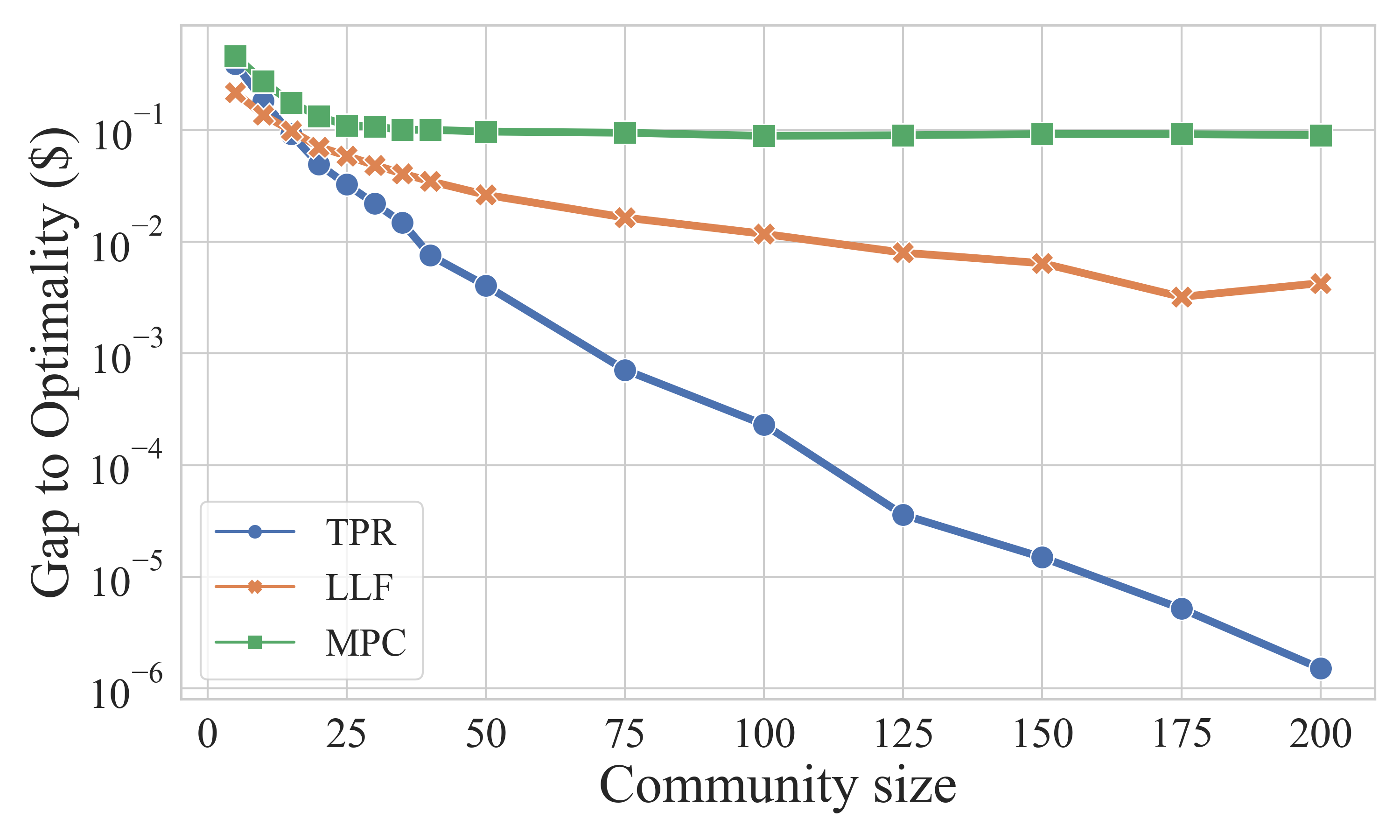}
 \caption{Optimality gap of centralized and distributed policies with $y$ axis plotted on a log scale. Individual renewable generation was drawn from a log-normal distribution with mean $\theta_g = 1.6$ kWh, i.i.d. across members and time. The EV arrival at an empty charger followed a Bernoulli process with probability $\alpha = \theta_g / (\bar c \bar T)$, with a maximum charging duration of 6 hours.}
 \label{fig:asymptoticopt}
\end{figure}

\subsection{Individual cost saving comparisons}
We evaluated individual cost saving inside under TPR over outside the community under NEM, validating the individual rationality of TPR. In addition, we evaluated alternative price-based distributed scheduling derived from the pricing rule by  proposed in  \cite{chakraborty2018analysis}, herein referred to as NEM {\it ex post} (NEM-exp). Specifically, NEM-exp is a pricing policy based on realized community net-consumption $z_t=\sum_{i \in [N]} z_{i,t}$ applied to realized individual net-consumption $z_{it}$ with the payment function
\begin{equation}\label{eq:chak_pricingrule}
 P_{\mbox{\tiny NEM-exp}} (z_{i,t}) = \begin{cases}
  \pi^+ z_{i,t}, & \sum_{j \in [N]} z_{j,t} > 0 \\
  \pi^- z_{i,t}, & \sum_{j \in [N]} z_{j,t} \le 0,
 \end{cases}
\end{equation}
To apply NEM-exp for distributed scheduling, we need a mechanism to obtain first the individual scheduling $z_{i,t}$, and compute payment based on the community level net-consumption. A simple strategy that also satisfies individual rationality and revenue adequacy is using NEM to schedule consumption and NEM-exp to price the consumption.

\begin{figure}[h]
 \centering
 \includegraphics[width=\columnwidth]{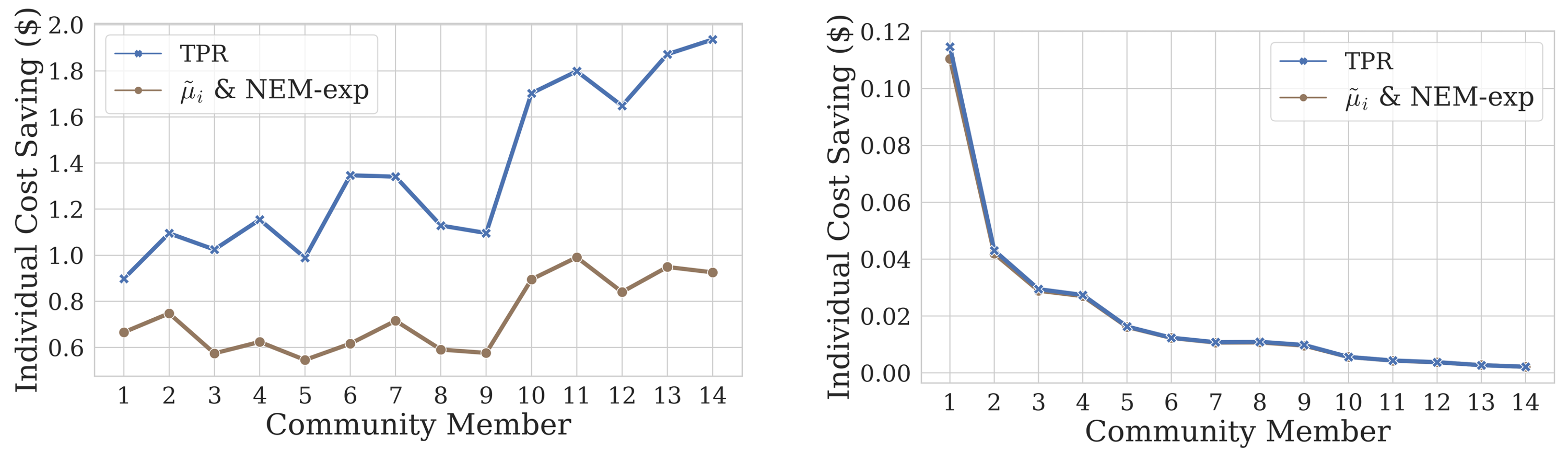}
 \caption{Individual cost saving under private rooftop solar (left) and shared community solar with fixed shares (right)}
 \label{fig:indcostsaving}
\end{figure}

We compared individual cost savings for $N = 14$ under two renewable configurations: private rooftop solar (Fig.~\ref{fig:indcostsaving} left with independent renewables, ) and shared community solar (Fig.~\ref{fig:indcostsaving} right with fixed shares of the common realization). Members were indexed in descending order of renewable generation capacity, with member 1 holding the largest share.

Under privately owned renewables, TPR achieved 33.74\%–67.58\% greater savings across members than $\tilde \mu_i$ with (\ref{eq:chak_pricingrule}). Decoupled renewable generation creates more opportunities to benefits from TPR, as individual net consumptions are more likely to have the opposite sign to the community net consumption, a scenario where households benefit by TPR. Savings were larger for members with small generation capacity (10–14), who rely more on purchased energy and thus benefited most from buying shared renewables at $\pim$ rather than importing at $\pip$. 

Under shared central renewables, TPR still outperformed (\ref{eq:chak_pricingrule}) for every member, but only by 0.38\%–3.83\%. Overall savings were smaller because coupled generation limits opportunities to exploit price differentials: surplus and deficit tend to occur simultaneously across members, so the two scenarios that benefit TPR, depicted in Fig.~\ref{fig:IR}, occurred less frequently.

%% file: sections/Conclusion_v8.tex
This paper develops a pricing mechanism for distributed EV scheduling within energy communities, formulated as a bilevel stochastic dynamic program. The main contribution is the threshold pricing rule (TPR), derived from the structure of the optimal centralized policy, which translates the value of aggregate renewables into simple price signals that members respond to individually. We prove that TPR satisfies individual rationality and revenue adequacy, and that the community cost under TPR converges to the centralized optimum as the community grows under light-traffic conditions. We further extend TPR to price-elastic, non-deferrable demand, showing that the same threshold structure preserves individual rationality with respect to household surplus.

A possible generalization of our result is to relax the price-taker assumption, which could open the door to no-regret online learning algorithms that adapt prices to members' strategic behavior. Another direction is to relax the assumption of truthful deadline revelation. Designing delay-differentiated pricing mechanisms to enforce truthful revelation is an interesting direction for future study. 

%% file: sections/appendixA_proof_v7.tex
Following lemma is used in the proofs. 

\begin{lemma}[No incompletion]\label{lem:completion}
	Suppose A1--A2 hold and each interval payment is convex piecewise linear with slopes in $[\pi^-,\pi^+]$ (true of $P_{\bm\pi}$ and of 	$P_{\bm\psi_t}$ under any TPR price). Then every optimal charging sequence---centralized or individual---serves each job fully by its
	deadline. 
\end{lemma}

\begin{proof}
	We prove that any policy with $c_{i,t}<m_{i,t}(d_{i,t},\tau_{i,t})$ at some state that can finish the charging is strictly improvable by an adapted 	modification. Then, the optimality of the completion follows. The argument is given for a single member facing prices with slopes in $[\pi^-,\pi^+]$ but the centralized community case is identical, for member by member, since raising $c_{i,t}$ raises $z_t$ and $P_{\bm\pi}$ obeys the same slope bounds.
	
	Suppose at interval $t$ the policy takes $c_{i,t}<m_{i,t}$, so $m_{i,t}>0$ and let $\epsilon:=m_{i,t}-c_{i,t}\in(0,\bar c]$. Consider a policy that modifies the action at this state only by charging $c_{i,t}+\epsilon$ at $t$ and leave all subsequent actions unchanged. Since future charging is at most $(\tau_{i,t}-1)\bar c$ and $m_{i,t}=d_{i,t}-(\tau_{i,t}-1)\bar c$, on every sample path the original policy ends the job with shortfall	$\delta \ge m_{i,t}-c_{i,t}=\epsilon$. On the other hand, the modified policy ends with shortfall $\delta - \epsilon \ge 0$. The payment increase at $t$ is at most  $\pi^+\epsilon$ by the slope bound. By convexity and A1, the penalty decreases by $q(\delta)-q(\delta-\epsilon)\ \ge\ q(\epsilon)-q(0)\ \ge\
	q'(0)\,\epsilon\ >\ \pi^+\epsilon$. The modification therefore strictly lowers cost on every path, so no optimal policy charges below $m_{i,t}$ at any reachable state. Finally, accompanied with A2, $c_{i,k} \ge m_{i,k}$ for all $k$ implies completion of the job by its deadline. 
\end{proof}

\subsection{Proof of Theorem 1}
By Lemma~\ref{lem:completion}, restrict to completing policies, so the stage cost is $P_{\bm\pi}(z_t)$. Fix $t$ and any completing policy
$\nu$. We show that $\nu$ can be modified to take the form (\ref{eq:cen_opt}) at interval $t$, without increasing the path-wise cost and without changing its actions before $t$. Applying these modifications at $t = 1, \ldots, T$ transforms $\nu$ into policy (\ref{eq:cen_opt}), proving its optimality. 

Fix a sample path and let $\bm c=\nu_t(\bm x_t)$, $z_t=\sum_i c_{i,t}-g_t$. Every modification below moves member $i$'s charge at $t$ by $\epsilon>0$ and offsets it at the earliest feasible interval(s) of the same job window $\{t+1,\dots,t+\tau_{i,t}-1\}$. This rule is adapted, and since both sequences complete the job, the state at the deadline is unchanged. Since $P_{\bm\pi}$ is convex piecewise linear with slopes
$\pi^-\le\pi^+$,
\begin{equation}
	\pi^-\epsilon \le P_{\bm\pi}(z+\epsilon)-P_{\bm\pi}(z)\le
	\pi^+\epsilon,\quad \forall z\in\mathbb R,\ \epsilon>0. \label{eq:21}
\end{equation}

\emph{Case 1 ($g_t\le m_t$).} Completion forces $c_{i,t}\ge m_{i,t}$ for all $i$, hence $z_t\ge m_t-g_t\ge0$. If $c_{i,t}>m_{i,t}$ for some
$i$, then $c_{i,t}>d_{i,t}-(\tau_{i,t}-1)\bar c$, so a future offset of $+\epsilon$ is feasible. Lowering $c_{i,t}$ keeps $z_t\ge0$ and saves
exactly $\pi^+\epsilon$ at $t$, while by (\ref{eq:21}) the offset raises future cost by at most $\pi^+\epsilon$. Iterating over members drives $\bm c_t$ to $(m_{1,t},\dots,m_{N,t})$ without increasing the cost, while the sign of $z_t$ is preserved at every step, so each exact-slope claim remains valid.

\emph{Case 2 ($g_t> M_t$).} Here $c_{i,t}\le M_{i,t}$ gives $z_t\le0$. If $c_{i,t}<M_{i,t}=\min\{d_{i,t},\bar c\}$, the job is unfinished at $t$, so $\nu$ charges a positive amount later; borrowing $\epsilon$ from the earliest such interval is feasible. Raising $c_{i,t}$ keeps $z_t\le0$, costing exactly $\pi^-\epsilon$ at $t$, while the future reduction saves at least $\pi^-\epsilon$ by (\ref{eq:21}). Driving $\bm c_t$ to $(M_{1,t},\dots,M_{N,t})$ does not increase cost.

\emph{Case 3 ($m_t< g_t\le M_t$).} If $z_t>0$, then $\sum_j c_{j,t}>g_t>m_t$, so some $i$ has $c_{i,t}>m_{i,t}$. Lowering it as in Case~1 (keeping $z_t\ge0$) saves exactly $\pi^+\epsilon$ against at most $\pi^+\epsilon$ later. If $z_t<0$, then $\sum_j c_{j,t}<g_t\le M_t$, so some $i$ has $c_{i,t}<M_{i,t}$. Raising it as in Case~2 (keeping $z_t\le0$) costs exactly $\pi^-\epsilon$ against a future saving of at least $\pi^-\epsilon$. Iterating drives $z_t$ to $0$, and the resulting allocation is feasible since $m_t<g_t\le M_t$.

Each inequality holds path-wise, hence in expectation, and all modifications are adapted, so (\ref{eq:cen_opt}) is an optimal policy. 
\subsection{Proof of Proposition 1}
Fix member $i$ and let $c_{i,t}^* := \mu_{i}^{\chitpr}(\pmb x_{i,t}, \pmb \psi_t)$ denote the action in the proposition. By A3, $\{\bm\psi_k\}$ is exogenous to $i$, so it suffices to show that, for \emph{every} realization of $\{g_{i,k},\bm\psi_k\}$, the $\bm c_{i}^*$ minimizes the realized cost of the current job over all feasible charging sequences. Path-wise optimality implies optimality in expectation over all adapted policies, and successive jobs decouple by Lemma~\ref{lem:completion}, so per-job optimality proves the proposition.

By Lemma~\ref{lem:completion}, restrict to completing sequences: $\sum_k c_{i,k}=d_{i,t}$, $c_{i,k}\in[0,\min\{d_{i,k},\bar c\}]$, with
realized cost $J(\bm c_i)=\sum_k P_{\bm\psi_k}(c_{i,k}-g_{i,k})$. Each $P_{\bm\psi_k}$ satisfies (\ref{eq:21}). Let $\bm c_i$ be any completing sequence. We modify $\bm c_i$ so that its $t$-th entry equals $c_{i,t}^*$, without increasing the realized cost $J$. Repeating these modifications at $t+1, t+2, \ldots$ transforms $\bm c_i$ into the $\bm c_i^*$, showing that $\bm c_i^*$ costs no more than $\bm c_i$. 

\emph{Case 1 ($\bm\psi_t=\bm\pi^+$).} Completion forces $c_{i,t}\ge m_{i,t}$. If $c_{i,t}>m_{i,t}$, move $\epsilon$ to the future (feasible as in Theorem~1): $P_{\bm\pi^+}$ has slope $\pi^+$ everywhere, so the saving at $t$ is exactly $\pi^+\epsilon$ and the future increase is at most $\pi^+\epsilon$ by (\ref{eq:21}).

\emph{Case 2 ($\bm\psi_t=\bm\pi^-$).} If $c_{i,t}<M_{i,t}$, completion implies a positive future charge and borrowing $\epsilon$ from it costs
exactly $\pi^-\epsilon$ at $t$ and saves at least $\pi^-\epsilon$ later.

\emph{Case 3 ($\bm\psi_t=\bm\pi$).} The $c_{i,t}^*$ is $\tilde\mu_{i,t}(\bm x_{i,t})$ of (\ref{eq:optstandalone}). If $c_{i,t}>\tilde\mu_{i,t}$, then $c_{i,t}>\max\{g_{i,t},m_{i,t}\}$, so $z_{i,t}>0$ and the local slope at $t$ is $\pi^+$: lower $c_{i,t}$ as in Case~1. If $c_{i,t}<\tilde\mu_{i,t}$, then $c_{i,t}<\min\{g_{i,t},M_{i,t}\}$, so $z_{i,t}<0$ and the local slope is $\pi^-$: raise $c_{i,t}$ as in Case~2. Either case preserves the sign of $z_{i,t}$ until the target is reached, so no step increases $J$.

Hence every completing sequence is dominated by $\mu_i^{\chi_{\text{TPR}}}$ on every realization, proving the proposition. \hfill$\square$

\subsection{Proof of Proposition~\ref{prop:IR}}
Individual rationality of $\chitpr$ is proved on a realization basis. Given a realization of renewables and EV arrivals, community charging prices over the scheduling horizon $(\pmb \psi_1^{\chitpr}, \ldots, \pmb \psi_T^{\chitpr})$ are uniquely determined. Let $\pmb c_i = (c_{i,1}, \ldots, c_{i, T})$ be member $i$'s optimal action sequence under TPR, and $\pmb c_i' = (c_{i,1}', \ldots, c_{i,T}')$ be the optimal stand-alone actions under NEM. 

We first show that the total charging cost of the optimal stand-alone action $\pmb c_i'$ under TPR is no greater than that under NEM for every $z_{i,t}' = c_{i,t}' - g_{i,t}$:
\begin{equation} \label{eq:IR_step1}
	\sum_{t=1}^T P_{\pmb \psi_t^{\chitpr}}(z_{i,t}' ) \le \sum_{t=1}^T P_{\pmb \pi}(z_{i,t}').
\end{equation}
For each interval $t$, there are three possible values of $\pmb \psi_t$ under TPR:
\begin{itemize}
	\item \textbf{Case 1} ($\pmb \psi_t = \pmb \pi^+ = (\pip, \pip)$, net consuming). If $c_{i,t}' > g_{i,t}$, $P_{\pmb \pi}  (z_{i,t}')  = P_{\pmb \psi_t^{\chitpr}}(z_{i,t}' )= \pip z_{i,t}'$ and if $c_{i,t}' \le g_{i,t}$, $P_{\pmb \pi}(z_{i,t}') = \pim z_{i,t}' >  P_{\pmb \psi_t^{\chitpr}}(z_{i,t}' )= \pip z_{i,t}'$. 
	\item \textbf{Case 2} ($\pmb \psi_t = \pmb \pi^- = (\pim, \pim)$, net producing). If $c_{i,t}' \le g_{i,t}$, $P_{\pmb \pi}  (z_{i,t}')  =P_{\pmb \psi_t^{\chitpr}}(z_{i,t}' ) = \pim z_{i,t}'$ and if $c_{i,t}' > g_{i,t}$, $P_{\pmb \pi}(z_{i,t}') = \pip z_{i,t}' >  P_{\pmb \psi_t^{\chitpr}}(z_{i,t}' ) = \pim z_{i,t}'$.
	\item \textbf{Case 3} ($\pmb \psi_t = \pmb \pi$). TPR and NEM are identical. 
\end{itemize}
This implies that (\ref{eq:IR_step1}) holds at every $t$ because TPR offers a price at least as favorable as NEM in every case. 

By Proposition~\ref{prop1}, $\pmb c_i'$ is suboptimal under TPR for every realization of $g_{i,t}$ and $\pmb \psi_t$, and therefore the following inequality holds for $z_{i,t} = c_{i,t} - g_{i,t}$:
\begin{equation}\label{eq:IR_step2}
	\sum_{t=1}^T P_{\pmb \psi_t^{\chitpr}}(z_{i,t} )\le \sum_{t=1}^T P_{\pmb \psi_t^{\chitpr}}(z_{i,t}' ).
\end{equation}
Combining (\ref{eq:IR_step1}) and (\ref{eq:IR_step2}), TPR satisfies individual rationality. 

\subsection{Proof of Theorem~2}
\label{app:thm-asymptotic}

Throughout the proof, both centralized optimal scheduling and the TPR scheduling are driven by the same realization of the primitive randomness, namely the arrival randomness, the duration--demand draws $(D_i, T_i)$, and the net renewables $\{g_{i,t}\}$. To make the arrival randomness explicit, associate with each household an i.i.d.\ coin sequence $\xi_{i,t} \sim \mathrm{Bernoulli}(\alpha)$, independent across households and intervals; an EV arrives at charger $i$ at the beginning of interval $t$ if and only if the charger is unoccupied and $\xi_{i,t} = 1$.

We first record two consequences of the arrival model. Define the
occupancy indicator and the occupancy count
\[
s_{i,t} := \mathbbm{1}\big\{ (d_{i,t}, \tau_{i,t}) \neq (0,0) \big\},
\qquad
I_N(t) := \sum_{i\in[N]} s_{i,t}.
\]
By the state dynamics, an EV arriving with duration $T_i$ occupies the charger for exactly $T_i$ intervals irrespective of the charging decisions, since $\tau_{i,t}$ decrements
deterministically and the departure occurs at the deadline. Hence occupancy is \emph{exogenous}: $s_{i,t}$ is a function of $\{\xi_{i,s}\}_{s \le t}$ and the duration draws alone, so that, under the coupling, the arrival times, the durations, the indicators
$s_{i,t}$, and the count $I_N(t)$ coincide on every sample path under
the two policies. Moreover, by the homogeneity of the arrival model,
$\{s_{i,t}\}_{i\in[N]}$ are i.i.d.\ across households with common mean
\begin{equation}
	\beta_t := \Pr(s_{i,t} = 1).
	\label{eq:occupancy-prob}
\end{equation}
Since occupancy at $t$ requires an arrival within the preceding $\min\{t, \mathcal T\}$ intervals, and every arrival requires the success of the exogenous coin in that interval,
\begin{equation}
	\beta_t \;\le\; 1 - (1-q)^{\min\{t,\mathcal T\}} \;\le\; q\,\mathcal T,
	\qquad \forall t \in [T].
	\label{eq:occupancy-bound}
\end{equation}

For each sample path of arrival and renewables, only the remaining-demand trajectories
$\{d_{i,t}\}$, and therefore the thresholds, may differ across the two policies. We write $\{m^{*}_{i,t}, M^{*}_{i,t}\}$ for the thresholds along the centralized optimal trajectory and
$\{m^{\text{TPR}}_{i,t}, M^{\text{TPR}}_{i,t}\}$ for those along the TPR trajectory; all lie in $[0, \bar{c}]$.

We will use the following policy-independent bound. Under any
charging policy, $d_{i,t} > 0$ only if $s_{i,t} = 1$, and
$0 \le m_{i,t} \le M_{i,t} = \min\{d_{i,t}, \bar{c}\}
\le \bar{c}\, s_{i,t}$. Summing over households,
\begin{equation}
	\sum_{i\in[N]} m_{i,t}
	\;\le\; \sum_{i\in[N]} M_{i,t}
	\;\le\; \bar{c}\, I_N(t),
	\label{eq:threshold-occupancy-bound}
\end{equation}
pathwise, for any policy. 
\begin{proof}
	\emph{Step 0 (Uniform integrable domination).}
	Under either policy, the per-house stage payment satisfies
	
	\begin{equation}	\label{eq:stage-domination}
		\begin{aligned}
			&\big|\frac 1 N P_{\pmb \pi}(z_t)\big| \le \pi^+ \frac 1 N |z_t| \\
			\le&  \pip \frac 1 N \sum_i (|c_{i,t}| + |g_{i,t}|) \le \pip (\bar c + \bar g_t^{(N)}),
		\end{aligned}
	\end{equation}
	where $\bar{g}^{(N)}_t := \frac{1}{N}\sum_{i\in[N]} g_{i,t}$. Since
	$\mathbb{E}[g_{i,t}] = \theta_g < \infty$, the right-hand side
	of~\eqref{eq:stage-domination} has finite expectation, uniformly in
	$N$. The deadline penalty is likewise uniformly bounded, since
	$d_{i,t} \le \bar{c}\mathcal T$ almost surely by Assumption~A2. Hence the
	per-house stage cost is dominated by an integrable random variable
	uniformly in $N$, under both policies.
	
	\emph{Step 1 (Zone probabilities).}
	For the centralized optimal trajectory, define the net-consuming,
	net-producing, and net-zero events
	\begin{align*}
		A^{N}_t &:= \Big\{ \sum_{i} m^{*}_{i,t} > g_t \Big\}, \quad
		B^{N}_t := \Big\{ \sum_{i} M^{*}_{i,t} < g_t \Big\} \\
		Z^{N}_t &:= \big( A^{N}_t \cup B^{N}_t \big)^{c},
	\end{align*}
	
	and define $A^{N,\text{TPR}}_t, B^{N,\text{TPR}}_t,
	Z^{N,\text{TPR}}_t$ analogously with the TPR thresholds. Because
	$\{s_{i,t}\}_{i\in[N]}$ are i.i.d.\ Bernoulli with mean $\beta_t$ and the
	renewables $\{g_{i,t}\}_{i\in[N]}$ are i.i.d.\ with mean $\theta_g$,
	the strong law of large numbers gives
	\begin{equation}
		\frac{1}{N} I_N(t) \xrightarrow{\;a.s.\;} \beta_t,
		\qquad
		\bar{g}^{(N)}_t \xrightarrow{\;a.s.\;} \theta_g.
		\label{eq:slln}
	\end{equation}
	By the bound~\eqref{eq:threshold-occupancy-bound}, which holds for
	both trajectories with the \emph{common} occupancy count $I_N(t)$,
	\begin{align*}
		A^{N}_t \cup A^{N,\text{TPR}}_t
		&\subseteq
		\Big\{ \tfrac{\bar{c}}{N} I_N(t) > \bar{g}^{(N)}_t \Big\},
		\\
		\Big\{ \tfrac{\bar{c}}{N} I_N(t) < \bar{g}^{(N)}_t \Big\}
		&\subseteq
		B^{N}_t \cap B^{N,\text{TPR}}_t.
	\end{align*}
	By~\eqref{eq:slln}, \eqref{eq:occupancy-bound}, and the light-traffic
	condition $ q < \theta_g / (\bar c \mathcal T)$
	\[
	\frac{\bar{c}}{N} I_N(t) - \bar{g}^{(N)}_t
	\xrightarrow{\;a.s.\;} \bar{c}\, \beta_t - \theta_g
	\;\le\; \bar{c}\, q \mathcal T - \theta_g \;<\; 0,
	\qquad \forall t \in [T].
	\]
	Therefore, for all $t \in [T]$,
	\begin{equation}
		\begin{aligned}
			&\Pr\!\big(A^{N}_t \cup A^{N,\text{TPR}}_t\big) \to 0,
			\qquad
			\Pr\!\big(B^{N}_t \cap B^{N,\text{TPR}}_t\big) \to 1,
			\\
			&\Pr\!\big(Z^{N}_t\big),\;
			\Pr\!\big(Z^{N,\text{TPR}}_t\big) \to 0.
		\end{aligned}
		\label{eq:zone-limits}
	\end{equation}
	\emph{Step 2 (Limiting per-house costs).}
	By Theorem~\ref{thrm1}, the per-house optimal stage cost
	decomposes over the three zones as
	\begin{equation}
		\begin{aligned}
			\frac 1 N \mathbb E[\omega(\pmb x_t, \pmb c_t^*)] = \Pr(\mathcal A_t^N)\pip \frac 1 N \mathbb E\big[\sum_i (m_{i,t}^* - g_{i,t}) | \mathcal A_t^N\big] \\
			+ \Pr(\mathcal Z_t^N) \cdot 0 + \Pr (\mathcal B_t^N) \frac 1 N \mathbb E\big[\sum_i (M_{i,t}^* - g_{i,t}) | \mathcal B_t^N\big],
		\end{aligned}
		\label{eq:cost-decomposition}
	\end{equation}
	where the net-zero term vanishes because the community neither imports
	nor exports and no deadline is missed under the optimal policy
	(Assumption~A1). All conditional terms
	in~\eqref{eq:cost-decomposition} are bounded in absolute value by the
	integrable dominating variable of Step~0. Combining the zone
	limits~\eqref{eq:zone-limits} with dominated convergence,
	\begin{equation}
		\lim_{N\to\infty} \frac{1}{N} V^{N}(x_0)
		= \sum_{t=1}^{T} \pi^{-}
		\lim_{N\to\infty} \frac{1}{N}
		\mathbb{E}\Big[ \sum_{i} \big( M^{*}_{i,t} - g_{i,t} \big) \Big],
		\label{eq:opt-limit}
	\end{equation}
	where the conditioning on $B^{N}_t$ is removed in the limit since
	$\Pr(B^{N}_t) \to 1$ and the conditional and unconditional
	expectations coincide asymptotically.
	
	Under $\chi_{\text{TPR}}$, Proposition~\ref{prop1}
	prescribes the individual actions $m^{\text{TPR}}_{i,t}$,
	$\mu_{\mathrm{p}}(x_{i,t})$, and $M^{\text{TPR}}_{i,t}$ in the three
	zones, and the per-house stage cost under TPR admits the analogous
	decomposition over $A^{N,\text{TPR}}_t$, $Z^{N,\text{TPR}}_t$, and
	$B^{N,\text{TPR}}_t$, with the net-zero term
	$\frac{1}{N}\mathbb{E}\big[ P_{\pi}\big( \sum_i (\mu_{\mathrm{p}}(x_{i,t})
	- g_{i,t}) \big) \,\big|\, Z^{N,\text{TPR}}_t \big]$ in place of zero.
	This term is bounded by the same dominating variable, and
	$\Pr(Z^{N,\text{TPR}}_t) \to 0$ by~\eqref{eq:zone-limits}, so it
	vanishes in the limit; the net-consuming term vanishes identically.
	Dominated convergence yields
	\begin{equation}
		\lim_{N\to\infty} \frac{1}{N} V^{N,\chi_{\text{TPR}}}(x_0)
		= \sum_{t=1}^{T} \pi^{-}
		\lim_{N\to\infty} \frac{1}{N}
		\mathbb{E}\Big[ \sum_{i} \big( M^{\text{TPR}}_{i,t} - g_{i,t} \big) \Big].
		\label{eq:tpr-limit}
	\end{equation}
	
	\emph{Step 3 (Trajectory coupling).}
	The renewables terms in~\eqref{eq:opt-limit} and~\eqref{eq:tpr-limit}
	are identical under the coupling, so it remains to show
	\[
	\Delta^{N}_t := \frac{1}{N}
	\mathbb{E}\Big[ \sum_{i} \big( M^{*}_{i,t} - M^{\text{TPR}}_{i,t} \big)
	\Big] \longrightarrow 0,
	\qquad \forall t \in [T].
	\]
	Define the good event
	$G_t := \bigcap_{s \le t} \big( B^{N}_s \cap B^{N,\text{TPR}}_s
	\big)$. On $G_t$, Theorem~\ref{thrm1} and
	Proposition~\ref{prop1} prescribe the same action
	$c_{i,s} = \min\{d_{i,s}, \bar{c}\}$ for every household and every
	$s \le t$ under the two policies. Since the arrivals, durations, and
	initial demands coincide under the coupling, an interval-by-interval
	induction gives $d^{*}_{i,s} = d^{\text{TPR}}_{i,s}$ for all $i$ and
	$s \le t$ on $G_t$; in particular, $M^{*}_{i,t} =
	M^{\text{TPR}}_{i,t}$ on $G_t$. Using $|M^{*}_{i,t} -
	M^{\text{TPR}}_{i,t}| \le \bar{c}$ on $G^{c}_t$, a union bound over
	the finitely many $s \le t$ together with~\eqref{eq:zone-limits} gives
	\[
	\big| \Delta^{N}_t \big|
	\;\le\; \bar{c} \, \Pr\!\big( G^{c}_t \big)
	\;\le\; \bar{c} \sum_{s \le t}
	\Pr\!\Big( \big( B^{N}_s \cap B^{N,\text{TPR}}_s \big)^{c} \Big)
	\;\longrightarrow\; 0.
	\]
	Substituting into~\eqref{eq:opt-limit} and~\eqref{eq:tpr-limit},
	\[
	\lim_{N\to\infty} \frac{1}{N}
	\big( V^{N,\chi_{\text{TPR}}}(x_0) - V^{N}(x_0) \big)
	= \sum_{t=1}^{T} \pi^{-} \lim_{N\to\infty} \Delta^{N}_t
	= 0.
	\]
\end{proof}

\subsection{Proof of Proposition 4}
Proof of surplus individual rationality is parallel to the proof of Proposition~\ref{prop:IR}. Let $\pmb c_i = (c_{i, 1}, \ldots, c_{i,T})$ be member $i$'s optimal charging action under $\tilde \chi_{\text{TPR}}$, and $\pmb c_i' = (c_{i,1}', \ldots, c_{i,T}')$ be the optimal stand-alone charging actions under NEM. Likewise, let $\pmb p_{i}' = (p_{i,1}, \ldots, p_{i,T})$ and $\pmb p_i = (p_{i,1}, \ldots, p_{i,T})$ be member $i$'s optimal non-deferrable demand schedule under NEM and TPR, respectively.

We first show  
\begin{equation}\label{eq:proofthrm4eq1}
	\begin{aligned}
		&\sum_{t = 1}^T U_{i,t}(p_{i,t}') - P_{\pmb \pi}(p_{i,t}' + c_{i,t}' - r_{i,t}) \\
		\le &\sum_{t = 1}^T U_{i,t}(p_{i,t}') - P_{\pmb \psi_t}(p_{i,t}' + c_{i,t}' - r_{i,t}). 
	\end{aligned}
\end{equation}

For each interval $t$, there are three possible values of $\pmb \psi_t$ under $\tilde \chi_{\text{TPR}}$. 
\begin{enumerate}
	\item Case 1 $\pmb \psi_t = \pmb \pi^-$. If $c_{i,t}' = m_{i,t}$, member $i$'s non-deferrable demand schedule is either $p_{i,t}(\pmb \pi^+)$ or $p_{i,t}(\pmb \pi, r_{i,t})$. Member $i$'s surplus under NEM is 
	\[
	U_{i,t}(p_{i,t}(\pmb \pi^+)) - \pip (p_{i,t}(\pmb \pi^+)  + m_{i,t} - r_{i,t})
	\]
	or
	\[
	U_{i,t}(p_{i,t}(\pmb \pi, r_{i,t})) - \pip(m_{i,t}). 
	\]
	For the same deferrable and non-deferrable demand decisions, member $i$'s surplus under $\tilde \chi_{\text{TPR}}$ are
	\[
	U_{i,t}(p_{i,t}(\pmb \pi^+)) - \pim (p_{i,t}(\pmb \pi^+)  + m_{i,t} - r_{i,t}) 
	\]
	or
	\[
	U_{i,t}(p_{i,t}(\pmb \pi, r_{i,t})) - \pim(m_{i,t}),
	\]
	which are both no less than the surplus under NEM. \\
	If the non-deferrable decision is $p_{i,t}(\pmb \pi^-)$, possible deferrable demand decision $c_{i,t}'$ is $m_{i,t}$, $r_{i,t} - p_{i,t}(\pmb \pi^-)$, or $M_{i,t}$. For each case, member $i$'s surplus under NEM is 
	\[
	U_{i,t}(p_{i,t}(\pmb \pi^-)) - \pip (p_{i,t}(\pmb \pi^-) + m_{i,t} - r_{i,t}) 
	\]
	or 
	\[
	U_{i,t}(p_{i,t}(\pmb \pi^-))
	\]
	or
	\[
	U_{i,t}(p_{i,t}(\pmb \pi^-)) - \pim (p_{i,t}(\pmb \pi^-) + M_{i,t} - r_{i,t}) .
	\]
	For the same deferrable and non-deferrable demand decisions, price under $\tilde \chi_{\text{TPR}}$ is no greater than that under NEM. Hence, member $i$'s surplus is no less than that under NEM. 
	
	\item Case 2 $\pmb \psi_t = \pmb \pi^+$. As in Case 1, if $c_{i,t}' = m_{i,t}$, member $i$'s non-deferrable schedule is either $p_{i,t}(\pmb \pi^+)$ or $p_{i,t}(\pmb \pi, r_{i,t})$. The payment of member $i$ under TPR and NEM is identical, as the price is $\pip$ for both cases. Hence, member $i$'s surplus does not change. \\
	
	If the non-deferrable demand schedule is $p_{i,t}(\pmb \pi^-)$, member $i$'s deferrable demand schedule $c_{i,t}'$ is $m_{i,t}$, $r_{i,t} - p_{i,t}(\pmb \pi^-)$, or $M_{i,t}$. For $c_{i,t}' = m_{i,t}$ and $r_{i,t} - p_{i,t}(\pmb \pi^-)$, member $i$'s surplus is identical as the payment under TPR and $\tilde \chi_{\text{TPR}}$ are the same. For $c_{i,t}'  = M_{i,t}$, member $i$'s is selling but at the higer price of $\pip$ under $\tilde \chi_{\text{TPR}}$. Thererfore, member $i$'s surplus is better off under $\tilde \chi_{\text{TPR}}$. 
	
	\item Case 3 $\pmb \psi_t = \pmb \pi$. Since member $i$'s payment is identical, member $i$'s surplus is identical. 
\end{enumerate}
In summary, for $\pmb p_i'$ and $\pmb c_i'$, $\tilde \chi_{\text{TPR}}$ always has price no worse off than the outside price and (\ref{eq:proofthrm4eq1}) holds.

Now, we prove the following inequality
\begin{equation}\label{eq:proofthrm4eq2}
	\begin{aligned}
		&\sum_{t = 1}^T U_{i,t}(p_{i,t}') - P_{\pmb \psi_t}(p_{i,t}' + c_{i,t}' - r_{i,t}) \\
		\le &\sum_{t = 1}^T U_{i,t}(p_{i,t}) - P_{\pmb \psi_t}(p_{i,t} + c_{i,t} - r_{i,t}).
	\end{aligned}
\end{equation}
For intervals with $\pmb \psi_t = \pmb \pi$, $p_{i,t} = p_{i,t}'$ as the optimal non-deferrable demand decision under $\tilde \chi_{\text{TPR}}$ matches the optimal decision under NEM. Since the net renewables for both prices are identical, $c_{i,t}= c_{i,t}'$. Hence, we only need to compare intervals with $\pmb \psi_t = \pmb \pi^+$ and $\pmb \pi^-$. 

For intervals with $\pmb \psi_t = \pmb \pi^{\pm}$ the payment function is linear,
\[
P_{\pmb \psi_t}(p_{i,t} + c_{i,t} - r_{i,t}) =\psi_t (p_{i,t} + c_{i,t} - r_{i,t}),
\] 
with $\psi_t = \pi^+$ or $\pim$. As $p_{i,t}$ is the optimal solution of (17), 
\[
U_{i,t}(p_{i,t}') - \psi_t (p_{i,t}') \le U_{i,t}(p_{i,t}) - \psi_t  (p_{i,t}). 
\]
For the charging actions $\{c_{i,t}'\}$ and $\{c_{i,t}\}$, the interchange argument in the Proposition~\ref{prop1} holds path-wise and its feasibility set, which implies 
\[
\sum_{t \in L}  \psi_t c_{i,t}' \ge \sum_{t \in L} \psi_t c_{i,t},
\] 
where $L = \{t \in [T] \; | \; \pmb \psi_t = \pmb \pi^+ \text{ or } \pmb \psi_t = \pmb \pi^-\}$. Combining two inequality, (\ref{eq:proofthrm4eq2}) holds. Therefore, $\tilde \chi_{\text{TPR}}$ satisfies surplus individual rationality. 

Proof for the revenue adequacy is parallel to the proof of Proposition~\ref{prop:RA} in the main text.

%% file: sections/appendixB_nondeferrable_load.tex
\subsection{Price-inelastic non-deferrable demand model}
	HVAC operation of member $i$ is modeled as a non-deferrable demand of the household, which is price-inleastic and controlled by the linear quadratic controller operating to track the reference temperature $\theta_{i,t}$. Specifically, following empirically verified linear dynamics of HVAC temperature control in \cite{bargiotas1988residential}
\begin{equation*}
	\begin{aligned}
		\min_{\{p_{i,t}\}_{t=1}^T} & \quad \mathbb E\bigg[ \sum_{t=1}^T (y_{i,t} - \theta_{i,t})^2\bigg] \\
		\text{s.t.} & \quad y_{i,t} = (1 - \gamma)y_{i,t-1} + \gamma u_{i,t} - \beta p_{i,t} + \omega _{i,t}, \quad \forall t,
	\end{aligned}
\end{equation*}
where $y_{i,t}$ and $u_{i,t}$ are indoor and outdoor temperature, respectively, $p_{i,t}$ is the energy consumption of HVAC and $\omega_{i,t}$ is the zero mean Gaussian disturbance. Here, indoor temperature is a state variable and we assume that it is perfectly measured. Two parameters $\gamma \in (0, 1)$ and $\beta$ are insulation coefficients and efficiency of the HVAC system, respectively. The optimal HVAC operation $p_{i,t}^*$ is an affine function of $y_{i,t-1}$:
\[
p_{i,t}^* = \frac{1}{\beta} \big(y_{i,t-1} + \gamma(\mathbb E[u_{i,t}] - y_{i,t-1}) - \theta_{i,t} \big),
\]
which can be derived by solving dynamic programming equation. 

Under assumptions that two random processes $u_{i,t}$ and $\omega_{i,t}$ are independent and are both time independent random variables, and $u_{i,t}$ are mutually independent across members, the optimal HVAC operation $\{p_{i,t}^*\}_{t=1}^T$ is also a time independent random process that is mutually independent across members.

\subsection{Asymptotic optimality of $\tilde \chi_{\text{TPR}}$}
We adopt the asymptotic regime of Sec.~IV-E: arrivals are homogeneous Bernoulli with probability $\alpha$ and maximum charging duration $\bar{T}$. Households are homogeneous in their preferences, $U_{i,t}=U$ for all $i$ and $t$, so that the price-elastic consumption levels $p(\pi^+)=\partial U^{-1}(\pi^+)$ and $p(\pi^-)=\partial U^{-1}(\pi^-)$ are deterministic constants common to all members. The renewables $\{r_{i,t}\}$ are i.i.d. across members and stationary in time with $\mathbb E[r_{i,t}]<\infty$. Let
\begin{equation}
	\theta_g(\pi^-) := \mathbb E\big[r_{i,t}\big]-p(\pi^-).
	\label{eq:thetag-elastic}
\end{equation}
Note that $\theta_g(\pi^-)$ is the mean surplus renewable for the non-deferrable demand consumption for price $\pim$, hence the most conservative measure of per-household surplus generation.

Let $W^N(\bm{x}_0)$ denote the optimal centralized welfare, in which
the coordinator jointly schedules $\{c_{i,t},p_{i,t}\}$ to maximize
the expected community surplus
\begin{equation}
	\mathbb E\Bigg[\sum_{t=1}^{T}\bigg(\sum_{i\in[N]} U_{i,t}(p_{i,t})
	- P_{\bm{\pi}}(z_t)
	- \sum_{j\in\mathcal{J}_t} q(d_{j,t}-c_{j,t})\bigg)\Bigg],
	\label{eq:welfare-obj}
\end{equation}
subject to (1)--(4), and let $W^{N,\tilde \chi_{\text{TPR}}}(\bm{x}_0)$ denote the expected community surplus under $\tilde \chi_{\text{TPR}}$ with the individually optimal responses in (\ref{eq:elastic-nd-threshold}) and Proposition~\ref{prop1}.

\begin{theorem}[Asymptotic optimality of $\tilde \chi_{\text{TPR}}$]
	\label{thm:elastic-asymptotic}
	Under the above regime, if $\alpha \le \theta_g(\pi^-)/(\bar{c}\bar{T})$,
	then
	\[
	\lim_{N\to\infty}\frac{1}{N}
	\Big(W^{N}(\bm{x}_0)-W^{N,\tilde \chi_{\text{TPR}}}(\bm{x}_0)\Big)=0.
	\]
\end{theorem}

\begin{proof}
	All policies are coupled to the same realization of arrivals, duration--demand draws, renewables, and utilities, as in the proof of	Theorem~\ref{thrm:asymptoticTPR}. By Lemma~1, whose slope hypothesis covers $P_{\bm{\pi}}$ and every $P_{\bm{\psi}_t}$ under $\tilde \chi_{\text{TPR}}$, we restrict attention to completing policies.
	
	Under any completing policy, the per-house stage surplus is bounded in absolute value by 
	\begin{equation}
		U\big(p(\pi^-)\big)
		+\pi^+\big(\bar{c}+p(\pi^-)+r_{i,t}\big),
		\label{eq:domination-elastic}
	\end{equation}
	where the first term is finite constant and the last has finite mean $\mathbb E[r_{i,t}]<\infty$. Averages over $i$ of \eqref{eq:domination-elastic} are thus uniformly integrable, justifying the interchanges of limits below.
	
	Since $P_{\bm{\pi}}(z)\ge \pi^- z$ for all $z\in\mathbb{R}$, replacing $P_{\bm{\pi}}$ by the linear payment $\pi^- z_t$ in \eqref{eq:welfare-obj} yields an upper bound $\bar{W}^N \ge W^N$. Under the linear price the relaxed problem decouples across households and intervals: each household consumes $p(\pi^-)$, and, by completion, each charging job pays exactly $\pi^- D_i$ regardless of its schedule. Hence $\bar{W}^N$ admits a	closed form free of any scheduling optimization:
	\begin{equation}
		\begin{aligned}
		\frac{1}{N}\bar{W}^{N}
		=&\sum_{t=1}^{T}
		\Big(U\big(p(\pi^-)\big)
		-\pi^-\big(p(\pi^-)-\mathbb E[r_{i,t}]\big)\Big) \\
		&-\frac{\pi^-}{N}\,\mathbb E\big[D^N\big],
		\label{eq:relaxed-value}
		\end{aligned}
	\end{equation}
	where $D^N$ is the total charging demand arriving over the horizon.
	
	The occupancy indicators $s_{i,t}$ and the count $I_N(t)$ are exogenous and identical across policies, and (\ref{eq:occupancy-bound})--(\ref{eq:threshold-occupancy-bound}) hold likewise. Since $p(\pi^-)$ is a common constant, the reported surplus levels	$\{g_{i,t}(\pi^-)=r_{i,t}-p(\pi^-)\}_{i\in[N]}$ are i.i.d. with mean $\theta_g(\pi^-)$. The strong law of large numbers together with (\ref{eq:occupancy-bound}) and the light-traffic condition $\alpha\le\theta_g(\pi^-)/(\bar{c}\bar{T})$ gives, for all $t\in[T]$,
	\[
	\frac{\bar{c}}{N}I_N(t)-\frac{1}{N}\sum_{i\in[N]}g_{i,t}(\pi^-)
	\xrightarrow{\text{a.s.}}
	\bar{c}\,\beta_t-\theta_g(\pi^-)<0.
	\]
	Hence, by (\ref{eq:elastic-tpr}) and (\ref{eq:threshold-occupancy-bound}),
	$\Pr\big(\tilde{B}^N_t\big)\to 1$, where
	\[
	\tilde{B}^N_t:=
	\Big\{\textstyle\sum_{i\in[N]}g_{i,t}(\pi^-)>M_t(\bm{d}_t)\Big\}
	\]
	is the event that $\tilde \chi_{\text{TPR}}$ broadcasts $\bm{\pi}^-$ at $t$.
	
	On event $G:=\bigcap_{t\le T}\tilde{B}^N_t$, the community price is $\bm{\pi}^-$ in every interval; each household consumes 	$p(\pi^-)$, charges $M_{i,t}(d_{i,t})$, completes every job, and every payment is linear at rate $\pi^-$. Hence, on event $G$, the realized community surplus under $\tilde \chi_{\text{TPR}}$ equals the realized optimal value of the relaxation path-wise. On $G^c$, the surplus gap is bounded by (\ref{eq:domination-elastic}), and
	\[
	\Pr(G^c)\le\sum_{t\le T}\Pr\big((\tilde{B}^N_t)^c\big)\to 0.
	\]
	Dominated convergence then yields
	\[
	\lim_{N\to\infty}\frac{1}{N}W^{N,\tilde \chi_{\text{TPR}}}
	=\lim_{N\to\infty}\frac{1}{N}\bar{W}^{N}.
	\]
	Combining with
	$W^{N,\tilde \chi_{\text{TPR}}}\le W^{N}\le \bar{W}^{N}$ completes the proof.
\end{proof}